\documentclass[11pt]{article}

\usepackage[utf8]{inputenc}
\usepackage[T1]{fontenc}

\usepackage{fullpage}
\usepackage{mdframed}
\usepackage{euscript}
\usepackage[dvipsnames]{xcolor}

\usepackage[
pagebackref,
colorlinks=true,
urlcolor=blue,
linkcolor=blue,
citecolor=OliveGreen,
]{hyperref}

\usepackage[a4paper]{geometry}

\usepackage{amsmath, amssymb, amsthm,verbatim, graphicx, xcolor,bbm}
\usepackage{braket}
\usepackage{MnSymbol}
\usepackage{tikz}
\usetikzlibrary{patterns}
\usepackage{doi}

\usetikzlibrary{decorations.pathreplacing}
\usepackage{stmaryrd} %\rrbracket
\usepackage{url}

\usepackage{braket}

\makeatletter
\newtheorem*{rep@theorem}{\rep@title}
\newcommand{\newreptheorem}[2]{%
\newenvironment{rep#1}[1]{%
 \def\rep@title{#2 \ref{##1}}%
 \begin{rep@theorem}}%
 {\end{rep@theorem}}}
\makeatother

\newtheorem{theo}{Theorem} %[section]
\newreptheorem{theo}{Theorem}

\newtheorem{lemma}[theo]{Lemma}

\newtheorem{definition}[theo]{Definition}

\newtheorem{prop}[theo]{Proposition}
\newtheorem{proposition}[theo]{Proposition}
\usepackage{thmtools}
\usepackage{thm-restate}

\def\cG{\mathcal{G}}

\def\1{\mathbbm{1}}

\def\eps{\varepsilon}
\newcommand{\beps}{\boldsymbol{\varepsilon}}
\newcommand{\error}{\boldsymbol{e}}

\DeclareUnicodeCharacter{00A0}{ }

\newcommand{\G}{\EuScript G}

\newcommand{\C}{\EuScript C}

\newcommand{\eQ}{\EuScript Q}
\newcommand{\Q}{Q}

\newcommand{\f}{\mathbb F}
\newcommand{\F}{\mathbb F}
\newcommand{\x}{\mathbf x}

\renewcommand{\geq}{\geqslant}
\renewcommand{\leq}{\leqslant}

\newcommand{\Cay}{{\sf Cay}}

\newcommand\DT[2]{C_{#1}\otimes\f^{#2}+\f^{#1}\otimes C_{#2}}

\newcommand{\col}{{\sf Col}}
\newcommand{\row}{{\sf Row}}

\begin{document}

\title{Efficient decoding up to a constant fraction of the code length for asymptotically good quantum codes}

\author{Anthony Leverrier\thanks{Inria, France. {\tt anthony.leverrier@inria.fr}} \qquad Gilles Z{\'e}mor\thanks{Institut de Math\'ematiques
de Bordeaux, UMR 5251, France. {\tt zemor@math.u-bordeaux.fr}}}

\date{\today}	

\maketitle

\begin{abstract}
We introduce and analyse an efficient decoder for quantum Tanner codes that can correct adversarial errors of linear weight. Previous decoders for quantum low-density parity-check codes could only handle adversarial errors of weight $O(\sqrt{n \log n})$. 
We also work on the link between quantum Tanner codes and the Lifted Product
codes of Panteleev and Kalachev, and show that our decoder can be adapted to the latter. The decoding algorithm alternates between sequential and parallel procedures and converges in linear time.  
\end{abstract}

%%%%%%%%%%%%%%%%%%%%%%%%%%%%%%%%%%%%

\section{Introduction}

\subsection{Contributions}
Historically, a major motivation behind the study of classical low-density parity-check (LDPC) codes was the possibility of efficient decoding. In his thesis, Gallager showed that random LDPC codes were asymptotically good codes with high probability and  proposed a first decoding algorithm based on message passing~\cite{Gal62}. Much later, Sipser and Spielman introduced expander codes which are explicit LDPC codes with a minimum distance linear in their length $n$, together with an efficient decoder that provably corrects adversarial errors of linear weight~\cite{SS96}. These expander codes are a special instance of Tanner codes~\cite{T81} defined on a $\Delta$-regular expander graph: bits are associated to the edges of the graph while parity-check constraints are enforced at each vertex \textit{via} small linear codes of length $\Delta$.

The LDPC property may be of even greater interest in the quantum case because such codes can significantly reduce the required overhead for fault-tolerant quantum computing~\cite{got14}. For this, it is enough to find a code family encoding $\Theta(n)$ logical qubits within $n$ physical qubits, with a sufficient minimum distance $d= \Omega(n^\alpha)$ for some $\alpha>0$ and an efficient decoding algorithm. Quantum expander codes, which are obtained by taking the hypergraph product \cite{TZ14} of two expander codes, form one such family when combined with the small-set-flip decoder~\cite{LTZ15, FGL18b}. This decoder corrects adversarial errors of weight proportional to the minimum distance, that is $O(\sqrt{n})$. Very recently, Panteleev and Kalachev discovered a first family of asymptotically good quantum LDPC codes~\cite{PK21} and suggested that a decoder similar to that of quantum expander codes might be able to correct errors of large weight in linear time.

In the present paper, we describe such a linear-time decoder for quantum Tanner
codes~\cite{LZ22}, a family of good quantum LDPC codes obtained by applying the Tanner code construction to a Left-Right Cayley complex~\cite{PK21,DEL21, BE21b} rather than to a graph. 
This decoder can correct adversarial errors of linear weight. Previously, the best decoder available was able to correct errors of weight $O(\sqrt{ n  \log n})$~\cite{EKZ20}.

\subsection{Context and history}

\paragraph{Quantum LDPC codes.}

Quantum error correcting codes (QECC), initially devised to fight decoherence with the goal of building large-scale quantum computers, have become a central object of study in fields as diverse as computer science, topological phases in physics, quantum information and even quantum gravity. Probably the best-known QECC is the toric code of Kitaev which corresponds to the degenerate ground-space of a simple local Hamiltonian defined on a torus~\cite{kit03}. It is a $2^2$-dimensional subspace encoding 2 qubits within the Hilbert space $(\mathbbm{C}^2)^{\otimes n}$ of $n$ physical qubits, defined as the common eigenspace with eigenvalue 1 of a set of $n-2$ commuting Pauli operators, called \emph{generators}. The \emph{minimum distance} of a code is the minimal weight\footnote{The weight of a Pauli operator is the number of qubits on which the operator acts nontrivially.} of a Pauli operator that commutes with all the generators but that cannot be written as a product of generators. The toric code achieves a remarkable distance of $\Theta(\sqrt{n})$, which remains the current record for quantum codes with generators acting locally in a topological space of constant dimension. 
In general, the parameters of such local codes are rather severely constrained by the ambient space~\cite{BT09, BPT10, del13, BK22}, and it is useful to relax the locality constraint and only require that the generators of the code have constant weight, but are otherwise arbitrary\footnote{We do not require the generators to be geometrically local anymore, and they can therefore act on an arbitrary set of qubits. This set should simply be of constant size.}. In addition, we ask that any qubit is involved in at most a constant number of generators. Such QECC are called LDPC.

An important stepping stone in the study of quantum LDPC codes was the
constant-rate\footnote{The rate of a QECC is the ratio $k/n$ of the number $k$
of logical qubits per physical qubit.} generalisation of the toric code called
hypergraph (or homological) product construction that forms a quantum LDPC code
from two arbitrary classical LDPC codes \cite{TZ14}. In particular, if the
classical codes are asymptotically good, then the resulting QECC has constant
rate and a minimum distance $d = \Theta(\sqrt{n})$. Ideas from
higher-dimensional expansion were useful to break the so-called \emph{square-root
barrier} for the minimum distance, but only by polylogarithmic factors
\cite{EKZ20, KT21}. A much more impressive improvement came from the idea of
adding a twist to the homological product construction~\cite{HHO20} to obtain $d
= \widetilde{\Theta}(n^{3/5})$. Further generalizations, either balanced product
codes or lifted product codes, were developed in~\cite{BE21b} and~\cite{PK20}
and finally led to the asymptotically good quantum LDPC codes of Panteleev and
Kalachev~\cite{PK21}. Hidden behind this final construction lies a
higher-dimensional generalization of a graph called a Left-Right Cayley complex
introduced in~\cite{DEL21} to construct (classical) locally testable codes.
Combining this complex, which is a balanced product of Cayley graphs, with the
Tanner code construction of~\cite{T81} yields an alternative family of good quantum LDPC codes called quantum Tanner codes~\cite{LZ22}. It should be emphasized that before these recent works, it was completely unclear whether quantum LDPC codes with a minimum distance significantly above $\sqrt{n}$ could exist at all.

\paragraph{Decoders for quantum LDPC codes.}

Topological codes, such as the toric code and its various generalizations in
higher dimensions, are by far the most studied quantum LDPC codes. In
particular, they now come with relatively efficient decoders that solve the
following problem: given the syndrome of a physical Pauli error $\error$,
\textit{i.e.}\ the list of generators that do not commute with the error, return
a guess $\hat{\error}$ for the error such that $\error$ and $\hat{\error}$
differ by a product of generators\footnote{In this case, $\error$ and
$\hat{\error}$ are called \emph{equivalent}.}, which is sufficient since generators act trivially on the codespace. This is a crucial relaxation compared to the classical decoding problem which requires to recover the exact error. While this could suggest at first sight that the decoding problem is simpler in the quantum case, this is in fact far from clear. A typical issue arises when the error corresponds to half a generator. In that case, it might not be clear for the decoder how to break the symmetry between the two equivalent errors corresponding to each half of the generator. A related observation is that the value of the error on a given qubit is never well defined since it is straightforward to find equivalent errors with a distinct action on that qubit.
These issues are by now relatively well understood for topological codes and the minimum weight perfect matching decoder~\cite{edm65} yields a decoder with very good performance and a reasonable complexity of $O(n^3)$. Faster decoders exist for topological codes, for instance Union-Find has essentially a linear complexity in the worst case and performs optimally for errors of weight less than $(d-1)/2$~\cite{DN21}.

Decoding more general, non topological LDPC codes, appears to be more challenging. While some solutions behave reasonably well against random errors, most notably the combination of Belief Propagation and Ordered Statistics Decoding proposed in~\cite{PK21b}, an important challenge is to understand how well one can correct adversarial errors. Inspired by a decoder of Hastings~\cite{has14} for 4-dimensional hyperbolic codes~\cite{GL14}, the small-set-flip (SSF) decoder of quantum expander codes can correct $\Theta(\sqrt{n})$ errors and works in a greedy fashion by trying to find small local patterns that can decrease the syndrome weight~\cite{LTZ15}. The decoder of~\cite{EKZ20} 
exploits the homological product structure in a more global way to (slightly) beat the $\sqrt{n}$ bound. 
The recent invention of good quantum LDPC codes raises the natural question of whether one can indeed correct adversarial errors of linear weight in polynomial, or even linear time\footnote{We note that~\cite{LH22b} studies such a decoder for a possible construction of asymptotically good quantum LDPC codes that relies on a conjecture about the existence of 2-sided lossless expander graphs with free group action.}. We answer this question in the affirmative, but note that our decoder is somewhat more complicated than SSF. In particular, while the \emph{normal} mode of the decoder is very similar to SSF, we also need to consider an \emph{exceptional} mode to take care of potentially problematic error patterns.\\

\paragraph{Comparison with recent results.}
Two independent works with decoders for linear weight errors have appeared at the same time as the present manuscript.
First Gu, Pattison and Tang present and analyse a related decoder for quantum Tanner codes~\cite{GPT22}. Instead of directly decoding the mismatch as here, they compute a global cost function equal to the sum of the weights of all the local corrections, and apply a greedy algorithm to find out which qubits to flip in order to decrease the value of the function. Technically, they improve the bound on the robustness of random codes, and show that there exists $\eps>0$ such that the dual of the tensor code obtained from random codes of length $\Delta$ is $\Delta^{3/2+\eps}$-robust with probability tending to 1 when the length $\Delta$ goes to infinity. This better robustness, strictly better than $\Delta^{3/2}$ allows for a simplified decoding algorithm, without any need for an exceptional mode as in the present paper. 
Then Dinur, Hsieh, Lin and Vidick present a variation of the codes of Panteleev
and Kalachev, where the roles of the qubits and of the generators are exchanged~\cite{DHL22}. One advantage is that the qubits are associated with 1-cells of
the chain complex while the generators are associated with 2-cells and 0-cells,
which corresponds to the standard framework of mapping a chain complex to a CSS
code. The hope is that this variation will generalise more easily to longer
chain complexes, which may be helpful to design elusive quantum locally testable
codes~\cite{AE15, has17, LLZ22}. Remarkably,~\cite{DHL22} establishes that dual
of tensor codes obtained from random codes satisfy optimal robustness, up to
$\Omega(\Delta^{2})$. This result was also obtained independently by 
Kalachev and Panteleev~\cite{PK22}. This better robustness then allows the authors to show that their quantum LDPC codes are also asymptotically good, similarly to the original PK codes and to quantum Tanner codes. Furthermore, it is also possible to extend the small set flip decoder of \cite{LTZ15} to design decoders for both kinds of errors that correct arbitrary errors up to a linear weight. 

Finally, the decoding algorithm introduced in the present paper can be significantly simplified if the local codes display optimal robustness. In that case, it becomes possible to parallelize the decoding procedure of both PK codes and quantum Tanner codes to obtain a logarithmic-time complexity~\cite{LZ22c}.\\

The paper is organised as follows. Section~\ref{sec:overview} gives an overview of the paper, recalls the construction of the quantum Tanner codes, describes the decoding algorithm and states our main result showing that the decoder corrects all error patterns of weight below a constant fraction of the code length. Section~\ref{sec:prelim} is a preliminary to the detailed part of the paper and introduces the required technical material. 
Section~\ref{sec:details} gives a detailed description of the quantum Tanner
codes  which rely here on a quadripartite version of the left-right Cayley
complex and not a bipartite version as in~\cite{LZ22}, since this would
complicate the exposition. Similarly, Section~\ref{sec:decoder} discusses 
the decoding algorithm and the tools used in its analysis. Section~\ref{sec:analysis} is the core
of the paper, giving the precise description of the decoder and its detailed
analysis, and establishing the
main theorem. Finally, Section~\ref{sec:LP} explores the links between the
quantum Tanner codes and the lifted product codes of Panteleev and Kalachev, and
explains how our decoding algorithm yields an efficient decoder for the lifted
product codes as well, thereby solving an open problem of \cite{PK21}.

\section{Overview}\label{sec:overview}

\paragraph{The left-right Cayley complex.} 
We will work with the square complex of \cite{PK21}, which can also be
thought of as a quadripartite version \cite{gol21} of the square complex of Dinur \textit{et al.}
\cite{DEL21}. We recall its construction, using the language of \cite{DEL21}.
It is an incidence structure $X$ between a set $V$ of vertices,
two sets of edges $E_A$ and $E_B$, that we will refer to as $A$-edges and
$B$-edges, and a set $Q$ of squares (or quadrangles). 
The vertex-set $V$ is defined from a group $G$. While we 
used a slightly more general bipartite version version of the complex
in~\cite{LZ22}, we prefer to use here the quadripartite version as it will
simplify the exposition. The vertex set is therefore partitioned as $V=V_{00}
\cup V_{01} \cup V_{10} \cup V_{11}$, with each part identified as a copy of the
group~$G$. Formally,
we set $V_{ij}=G\times\{ij\}$ for $i, j \in \{0,1\}$. We also have two self-inverse
subsets $A=A^{-1}$ and $B=B^{-1}$ of the group $G$: for $i \in \{0,1\}$, a vertex $v=(g,i0)\in V_{i0}$ and a vertex
$v'=(g',i1) \in V_{i1}$ are said to be related by an $A$-edge if $g'=ag$ for some $a\in A$.
Similarly, for $j \in \{0,1\}$, vertices $v = (g,0j)$ and $v' = (gb,1j)$ are said to be related by a $B$-edge if $g'=gb$ for
some $b\in B$. The sets $E_A$ and $E_B$ make up the set of $A$-edges and
$B$-edges respectively, and define graphs $\G_A$ and $\G_B$ where $\G_A$
consists of two copies of the double cover of the {\em left} Cayley graph
$\Cay(G,A)$ and $\G_B$ consists of two copies of the double cover of the  {\em right} Cayley graph  $\Cay(G,B)$.

Next, the set $Q$ of squares is defined as the set of $4$-subsets of vertices of the form
\[
\{(g,00),(ag,01),(gb,10),(agb,11)\}.
\]
The four vertices of a square threrefore belong to the four distinct copies of $G$, as depicted on Figure~\ref{fig:square}.

An advantage of the quadripartite version we consider here is that we do not need to enforce any additional constraint on $G, A, B$ such as the Total No-Conjugacy condition defined in~\cite{DEL21} requiring that $ag \ne gb$ for all choices of $g, a, b$.

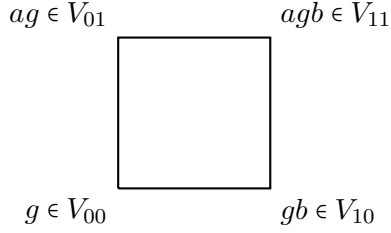
\begin{figure}[h!]
\begin{center}
\begin{tikzpicture}
\draw[thick] (0,0) -- (0,2) -- (2,2) -- (2,0) -- (0,0);
\node[below left] at (0,0) {$g \in V_{00}$};
\node[above left] at (0,2) {$ag \in V_{01}$};
\node[above right] at (2,2) {$agb \in V_{11}$};
\node[below right] at (2,0) {$gb \in V_{10}$};
\end{tikzpicture}
\end{center}
\caption{Square $\{ (g,00), (ag, 01), (agb,11), (gb,10)\}$ of the complex.\label{fig:square}
}
\end{figure}

If we restrict the vertex set to $V_0 := V_{00} \cup V_{11}$, every square is now
incident to only two vertices: one in $V_{00}$ and one in $V_{11}$. The set of squares can then be
seen as a set of edges on $V_0$, and it therefore defines a bipartite graph that we denote
by $\G_0^{\square}=(V_0,Q)$. Similarly, the restriction to the vertices of $V_1 := V_{01} \cup V_{10}$ defines the
graph $\G_1^{\square}$, which is an exact replica of $\G_0^{\square}$: both graphs are 
defined over two copies of the group $G$, with $g,g,'\in G$ being related by an edge
whenever $g'=agb$ for some $a\in A, b\in B$. We assume
for simplicity that $A$ and $B$ are of the same cardinality $\Delta$.
For any vertex $v$, we denote by $Q(v)$ the $Q$-neighbourhood of $v$
which is defined as the set of squares incident to $v$. The $Q$-neighbourhood
$Q(v)$ has cardinality $\Delta^2$ and is isomorphic to the product set $A\times
B$: the situation is illustrated on Figure~\ref{fig:code} and discussed in
detail in Section~\ref{sec:details}.

\begin{figure}[h]
\begin{center}
\begin{tikzpicture}
\draw (0,0) rectangle (3,3);
\draw[pattern=north west lines,pattern color = blue] (1,0) rectangle (1.5,3);
\draw[pattern=north east lines,pattern color = red] (0,1.5) rectangle (3,2);
\draw[step=0.5cm] (0,0) grid (3,3);

\draw (4,0) rectangle (7,3);
\draw[pattern=north west lines,pattern color = blue] (5,0) rectangle (5.5,3);
\draw[pattern=north east lines,pattern color = red] (5,1.5) rectangle (5.5,2);

\draw[step=0.5cm] (4,0) grid (7,3);

\draw (0,4) rectangle (3,7);
\draw[pattern=north east lines,pattern color = red] (0,5.5) rectangle (3,6);
\draw[pattern=north west lines,pattern color = blue] (1,5.5) rectangle (1.5,6);
\draw[step=0.5cm] (0,4) grid (3,7);

\draw (4,4) rectangle (7,7);
\draw[step=0.5cm] (4,4) grid (7,7);
\draw[pattern=north east lines,pattern color = red] (5,5.5) rectangle (5.5,6);
\draw[pattern=north west lines,pattern color = blue] (5,5.5) rectangle (5.5,6);

\node at (-1.4,0.5) {$Q(g,00)$};
\node at (8.5,0.5) {$Q(gb,10)$};
\node at (-1.4,6.5) {$Q(ag,01)$};
\node at (8.5,6.5) {$Q(agb,11)$};
\node at (1.25,-0.4) {$b$};
\node at (5.35,-0.35) {$b$};
\node at (-0.4,1.75) {$a$};
\node at (-0.4,5.75) {$a$};

\end{tikzpicture}
\end{center}
\caption{The four local views $Q(v)$ that contain the square $\{ (g,00), (ag, 01), (agb,11), (gb,10)\}$. The views of two vertices connected by an $A$-edge (resp.~a $B$-edge) share a row depicted in red (resp.~a column in blue). The labeling is chosen to ensure that a given square, such as the one in red and blue, is indexed similarly, by $(a,b)$ here, in the four local views. The $\sigma_X$-type generators are codewords of $C_A\otimes C_B$ in the local views of $V_{00} \cup V_{11}$; the $\sigma_Z$-type generators are codewords of $C_A^\perp \otimes C_B^\perp$ in the local views of $V_{01} \cup V_{10}$. They automatically commute since their support can only intersect on a shared row or column (as depicted), and the orthogonality of the local codes ensure that they commute on this row or column.
\label{fig:code}}
\end{figure}
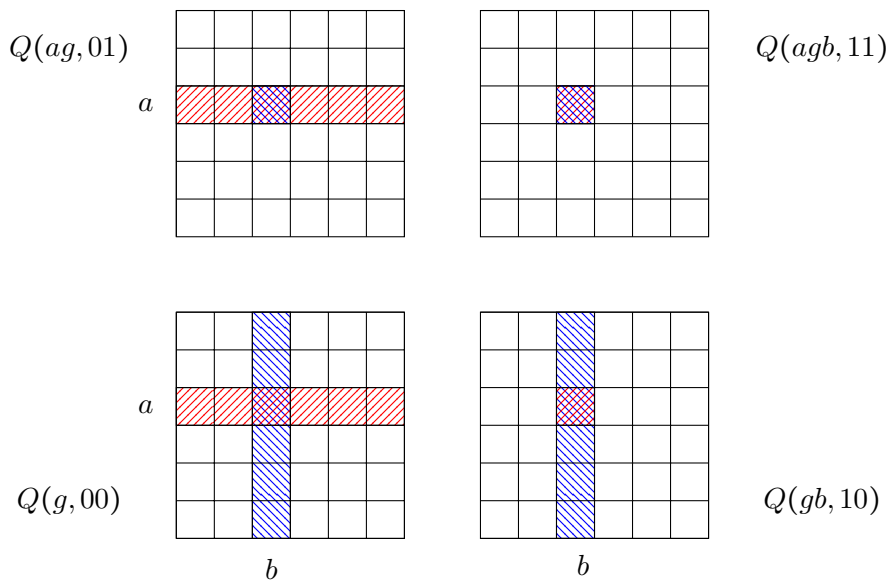

\paragraph{Quantum Tanner codes on the complex $X$.}
A \emph{Tanner code}, or expander code, on a $\Delta$-regular graph $\G=(V,E)$ is the
set of binary vectors indexed by $E$ (functions from $E$ to $\f_2$), such that
on the edge neighbourhood of every vertex $v\in V$, we see a codeword of a small code $C$ of length $\Delta$~\cite{T81,SS96}. We denote the resulting code by $\text{Tan}(\G, C) \subset \F_2^E$.

Following~\cite{LZ22}, we consider \emph{quantum Tanner codes} which are quantum CSS codes formed by two classical Tanner codes $\C_0$ and $\C_1$ with support on the set $Q$ of squares of a Left-Right Cayley complex. The CSS construction requires both codes to satisfy the orthogonality condition $\C_0^\perp \subset \C_1$. Enforcing this condition requires some care for the choice of the local codes of $\C_0$ and $\C_1$. We will define local codes on the space $\f_2^{A\times B}$ that we may think of as the space of matrices whose rows (columns) are indexed by $A$ (by $B$).
If $C_A\subset \f_2^A$ and $C_B\subset\f_2^B$ are two linear codes,
we define the {\em tensor} (or product) code $C_A\otimes C_B$ as the space of
matrices $x$ such that for every $b\in B$ the column vector $(x_{ab})_{a\in A}$
belongs to $C_A$ and for every $a\in A$ the row vector $(x_{ab})_{b\in B}$
belongs to $C_B$.
Recall that the dual $C^\perp$ of a code $C$ of length $n$ is the set of words orthogonal to all words in $C$:
\[ C^\perp := \{ x \in \F_2^n \: : \: \langle x, y \rangle = 0 \; \forall y \in C\}.\]
We finally define $\C_0$ and $\C_1$ to be the following classical Tanner codes:
\[ \C_0 = \text{Tan}(\G_0^\square,(C_A \otimes C_B)^\perp), \quad \C_1 = \text{Tan}(\G_1^\square,(C_A^\perp \otimes C_B^\perp)^\perp),\]
with bits associated to each square of $\Q$ and local constraints enforced at the vertices of $V_0$ and $V_1$, respectively. We will refer to the dual of tensor codes $(C_A \otimes C_B)^\perp$ as \emph{dual tensor codes}.
Let us denote $C_0 := C_A \otimes C_B$ and $C_1 := C_A^\perp \otimes C_B^\perp$, so that $\C_i = \text{Tan}(\G_i^\square, C_i^\perp)$ for $i \in \{0,1\}$. 
To check the orthogonality condition between the two codes, it is convenient to look at their generators (or parity-checks). We define a $C_0$-generator for $\C_0$ (resp.~a $C_1$-generator for $\C_1$) as a vector of $\F_2^\Q$ whose support lies entirely in the $\Q$-neighborhood $\Q(v)$ of $V_{0}$ (resp.~$V_1$), and which is equal to a codeword of $C_0$ (resp.~$C_1$) on $\Q(v)$. The Tanner code $\C_0$ (resp.~$\C_1$) is defined as the set of vectors orthogonal to all $C_0$-generators (resp.~$C_1$-generators).
The condition $\C_0^\perp \subset \C_1$ simply says that all $C_0$-generators
are orthogonal to all $C_1$-generators: this follows from the fact that if a
$C_0$-generator on $v_0 \in V_0$ and a $C_1$-generator on $v_1 \in V_1$ have
intersecting supports, then $v_0$ and $v_1$ must be neighbours in the left-right Cayley complex and their local views must intersect on either a column or a row, on which the two generators equal codewords of $C_A$ and $C_A^\perp$, or of $C_B$ and $C_B^\perp$ (see Fig.~\ref{fig:code}).

We denote by $\eQ = (\C_0, \C_1)$ the \emph{quantum Tanner code} obtained in this way from $\C_0$ and $\C_1$, and recall the main result from~\cite{LZ22}.

\begin{theo}[Quantum Tanner codes are asymptotically good~\cite{LZ22}]\label{thm:good-qLDPC}
There exists an infinite family of square complexes $X$ such that the following holds.  
For any $\rho \in (0,1/2)$, $\eps\in(0,1/2)$ and $\delta>0$ satisfying
 $-\delta \log_2 \delta - (1-\delta)\log_2(1-\delta) < \rho,$
randomly choosing $C_A$ and $C_B$ of rates
$\rho$ and $1-\rho$ yields, with probability $>0$ for $\Delta$ large enough,
an infinite sequence of quantum codes $\eQ=(\C_0,\C_1)$ of rate $(2\rho-1)^2$,
length $n$ 
and minimum distance $\geq \delta n/4 \Delta^{3/2+\eps}$.
\end{theo}

Concretely, the complexes $X$ that work in Theorem~\ref{thm:good-qLDPC} are all
complexes such that $\Cay(G,A)$ and $\Cay(G,B)$ are Ramanujan graphs and for
which $\Delta$ is fixed, independent of the complex size, and large enough.

The proof of Theorem~\ref{thm:good-qLDPC} in~\cite{LZ22} went along the
following lines: If one considers a codeword $\mathbf{x}$ of $\C_1^\square$ and
identifies it with its support, a set of squares that reduces to a set of edges in the graph $\G_1^\square$, then
this is a subgraph of minimum degree at least $\delta\Delta$, a lower bound on
the minimum distance of the component dual tensor code. Expansion in $\G_1^\square$ is not
quite enough to deduce that $|\x|$ must be large, but if $|\x|$ is small,
expansion does tell us that its
local views on vertices of $V_1$ must have a small weight (close to $\Delta$)
on average. If we decompose these local views, which have the structure of dual
tensor codewords on $A\times B$, as sums of column vectors, \textit{i.e.}\ elements of
$C_A\otimes\F_2^B$, and row vectors in $\F_2^A\otimes C_B$, then,
the definition of the complex $X$ tells us that the individual column
and row vectors from these local views also exist in the $Q$-neighbourhoods of
vertices of $V_0$.
Switching to expansion in the graphs $\G_A$ and $\G_B$, we obtain that most
individual column $C_A$-codewords and row $C_B$-codewords from the
$Q$-neighbourhoods of vertices of $V_1$, must cluster around local views of
vertices of $V_0$. A local analysis of these clustered local views then shows that
adding some non-zero tensor codeword must reduce its Hamming weight, and in this
way one obtains iteratively that the global codeword $\x$ of $\C_1$ can be
expressed as a sum of generators, which yields the lower bound on the minimum
distance.

\paragraph{Robustness of the component codes $C_A$ and $C_B$.} Crucial to the
analysis sketched above is the ability to claim that if a dual tensor codeword
$x=(x_{ab})\in \F_2^{A\times B}$ has sufficiently small weight, then it can be expressed as a sum
$x=r+c$, where the union of row codewords $r\in \F_2^A\otimes C_B$ and the union
of column codewords $c\in C_A\otimes\F_2^B$ are both of small weight. This is a
{\em robustness property.} More precisely, we say that a dual tensor code $(C_A^\perp\otimes
C_B^\perp)^\perp=C_A\otimes\f_2^B + \f_2^A\otimes C_B$ is
$w$-robust if any codeword $x$ of weight $\leq w$ has its support
included in the union of $|x|/d_A$ columns and $|x|/d_B$ rows, where $d_A$ and
$d_B$ are the minimum distances of $C_A$ and $C_B$. A similar notion
is used both in~\cite{DEL21} and~\cite{PK21}. What was shown in~\cite{LZ22} is that 
for any $\eps >0$ and large enough $\Delta$, 
when $C_A$ and $C_B$ are chosen at random, then the dual tensor codes 
$(C_A\otimes C_B)^\perp$ and $(C_A^\perp\otimes C_B^\perp)^\perp$
are both $\Delta^{3/2-\eps}$-robust with high probability. This property will
again be crucial when considering decoding issues.

\paragraph{The decoding problem.} 
A standard property of QECC is that the ability to correct all Pauli errors of weight up to $t$ implies the ability to correct \emph{arbitrary} errors of weight less than $t$. For a CSS code, bit flips and phase flips can be decoded independently, and it is therefore sufficient to consider only one type of errors. Since we chose $C_A$ and $C_B$ of rate $\rho$ and $1-\rho$, respectively, we see that the classical Tanner codes $\C_0$ and $\C_1$ have the same parameters and that the resulting quantum Tanner codes will protect equally well against bit flips and phase flips. 
Without loss of generality, we therefore consider a phase-flip error with
support on $\error \in \F_2^\Q$. This error is detected by the classical Tanner
code $\C_1$, and the goal of the decoder is to output an error candidate
$\hat{\error}$ such that $\error + \hat{\error} \in \C_0^\perp$. Recall indeed
that an element of $\C_0^\perp$ is a sum of generators, and therefore acts
trivially on the codespace. For this reason, it is not necessary to recover the
exact error $\error$, and any \emph{equivalent} error in the coset $\error + \C_0^\perp$ is equally good.

An approach mentioned in~\cite{LZ22} for decoding (classical) Tanner codes is to
define a \emph{mismatch} vector that summarises how the local decoders associated to each local code may disagree about the error, and then try to locally modify this mismatch in order to reduce its weight. 
It is natural to see the error $\error$ as a collection of local views on the
vertices of $V_1$: abusing notation slightly, we can write $\error = \{e_v\}_{v
\in V_1}$, where we see the local views $e_v$ both as small
length vectors in $\F_2^{Q(v)}$ and as vectors of $\F_2^Q$ with $0$ coordinate
values
outside $Q(v)$. Since each square of $Q$ belongs both to a local view of
$V_{01}$ and to a local view of $V_{10}$, we have that $\sum_{v \in V_{01}} e_v = \sum_{v \in V_{10}} e_v$. 
For each vertex $v \in V_1$, one can compute a local error $\eps_v$ with support on $Q(v)$ of minimal Hamming weight yielding the corresponding local syndrome. This gives a decomposition of the local views of the error 
\[e_v = \eps_v + c_v + r_v,\]
with $c_v \in C_A\otimes \F_2^B, r_v = \F_2^A \otimes C_B$, and $\eps_v$ of minimal Hamming weight. Here, we have that $c_v + r_v$ is a codeword of the dual tensor code $C_1^\perp = C_A\otimes \F_2^B +  \F_2^A \otimes C_B$. 
The issue is that the local views $\{\eps_v\}_{v\in V_1}$ are in general not consistent and do not define a global error candidate. 
We measure this inconsistency by defining the mismatch vector 
\begin{equation}
\label{eq:mismatch}
Z := \sum_{v \in V_1} \eps_v\in \F_2^Q.
\end{equation}
If it is equal to zero, it means that each square/qubit is affected the same value for the two views it belong to, and the decoder is able to define a global error. Otherwise, the support of $Z$ corresponds to the set of squares for which the local views disagree. Exploiting the previous remark that $\sum_{v \in V_1} e_v = 0$, we can rewrite the mismatch as
\[
Z = \sum_{v \in V_1} r_v +c_v = C_0 + R_0 + C_1+R_1,
\]
where we defined
\[ C_0 = \sum_{v \in V_{10}} c_v, \quad R_0 = \sum_{v \in V_{01}} r_v, \quad C_1 = \sum_{v \in V_{01}} c_v, \quad R_1 = \sum_{v \in V_{10}} r_v.\]

The idea behind our decoder is to find a decomposition $\{\hat{r}_v, \hat{c}_v\}_{v \in V_1}$ such that $Z =  \sum_{v \in V_1} \hat{r}_v +\hat{c}_v$. 
In that case, the decoder will output the error candidate $\hat{\error}=\{\hat{e}_v\}_{v \in V_1}$ with 
\begin{equation}\label{eq:evhat}
 \hat{e}_v := \eps_v + \hat{r}_v + \hat{c}_v.
\end{equation}
In particular, the vectors $\error$ and $\hat{\error}$ differ by an element of $\C_1$
(since they have the same syndrome on $V_1$), and a sufficient condition to
guarantee the success of the decoder is that $|\error + \hat{\error}|$ is less than the
minimum distance. In that case, it means that $\error$ and $\hat{\error}$ necessarily
differ by an element of $\C_0^\perp$, that is a sum of generators. It is not
difficult to see that $|Z| = O(|\error|)$. Therefore if $|\error| \leq \kappa n$ for
some sufficiently small $\kappa>0$ and if the algorithm can find $\hat{\error}$ of weight at most $O(|Z|)$, then the decoder will return a correct solution since the minimum distance is linear in $n$. 

We now describe the main subroutine of decoding algorithm, which aims at finding a decomposition $Z =  \sum_{v \in V_1} \hat{r}_v +\hat{c}_v$.
We will keep track of 5 vectors initialized as follows
\[ \hat{Z} := Z, \quad \hat{C}_0 := 0, \quad \hat{R}_0:=0, \quad \hat{C}_1 :=0, \quad \hat{R}_1 :=0.\]
It will alternate between two procedures:

\paragraph{Sequential procedure in $V_0$.}
While there exists some $v\in V_{ii}$ and $c_v \in C_A \otimes \F_2^B, r_v \in  \F_2^A \otimes C_B$ such that $|\hat{Z}+ c_v +r_v| < |\hat{Z}|$, perform the update:
\[ \hat{Z} \leftarrow \hat{Z} + c_v +r_v, \quad \hat{C}_i \leftarrow \hat{C}_i +c_v, \quad \hat{R}_i\leftarrow \hat{R}_i+r_v.\]

Now it may happen that this sequential decoder will get stalled at some point,
and not be able to decrease the Hamming weight $|\hat{Z}|$ by a local
modification on a $Q$-neighbourhood of a vertex of $V_0$. When this happens, we
can consider the subgraph of $\G_0^\square$ induced by the vector $(R_0+C_0)\cap
(R_1+C_1)$ and realise that its minimum degree must be at least
$\delta\Delta/2$ (provided the initial error $\error$ is of sufficently small
weight). Conceptually, this subgraph does not look very different from
a Tanner codeword, and if we follow the blueprint of the proof of the lower
bound on the minimum distance (sketched just after
Theorem~\ref{thm:good-qLDPC}), then it is natural to expect that the stalled 
decoder will be unlocked simply by switching the sequential decoding procedure
to vertices of $V_1$. However, we cannot quite make this work. The issue is that
we again need robustness of the component codes $C_A,C_B$, and this time the
robustness parameter $w=\Delta^{1/2-\eps}$ guaranteed us by random choice falls
just short. To circumvent this problem, we use a more complicated decoding
procedure to unlock the stalled sequential decoder. It consists of two rounds of
parallel decoding.

\paragraph{Parallel decoding procedure.}
\begin{itemize}
\item[--] \textbf{First parallel decoding step, on vertices of $V_1$.} 
Identify vertices of $V_1$ for which there exists $c_v \in C_A \otimes \F_2^B,
r_v \in  \F_2^A \otimes C_B$ and  subsets $A_0\subset A$ and $B_0\subset B$ of
the local coordinate sets that are sufficiently large and 
such that the Hamming weight of $\hat{Z}+ c_v +r_v$ decreases sufficiently
{\em on the reduced coordinate set} $A_0\times B_0$. 
For all vertices $v$ where this is satisfied, update as before:
\[ \hat{Z} \leftarrow \hat{Z} + c_v +r_v, \quad \hat{C}_j \leftarrow \hat{C}_j
+c_v, \quad \hat{R}_i\leftarrow \hat{R}_i+r_v.\]
See Section~\ref{sec:tweaked} for a precise description of the criteria for
updating.
\item[--]
\textbf{Second parallel decoding step, on vertices of $V_0$.}
Identify vertices of $V_0$ for for which there exists $c_v \in C_A \otimes \F_2^B,
r_v \in  \F_2^A \otimes C_B$ such that $|\hat{Z}+ c_v +r_v|<|\hat{Z}|$. Among
all possible choices, maximise the difference $|\hat{Z}|-|\hat{Z}+ c_v +r_v|$.
For all such vertices perform the update
\[ \hat{Z} \leftarrow \hat{Z} + c_v +r_v, \quad \hat{C}_j \leftarrow \hat{C}_j
+c_v, \quad \hat{R}_i\leftarrow \hat{R}_i+r_v.\]
\end{itemize}

What the decoder does by default is apply the sequential decoding procedure.
Whenever this is not possible, it applies the parallel decoding procedure.
It continues until $\hat{Z} = 0$, at which point it stops and outputs
$(\hat{R}_0, \hat{C}_0, \hat{R}_1, \hat{C}_1)$.

A word of comment is in order here. As mentioned above, what happens is that 
when the sequential decoder is stalled, we 
cannot guarantee the existence of a $Q$-neighbourhood $Q(v)$, for $v\in V_1$, on
which we can decrease $|\hat{Z}|$. But this is almost the case: what we can
guarantee is
the existence of vertices $v$ for which one can decrease
$|\hat{Z}|$ on a large subset of $Q(v)$. Furthermore, this will be the case for
most vertices of $V_1$ on whose $Q$-neighbourhoods $\hat{Z}$ has a large Hamming
weight: this is why we apply this tweaked local decoding procedure in parallel
on all possible vertices of $V_1$. After this first parallel decoding step we
cannot guarantee that $|\hat{Z}|$ has decreased, however we can guarantee a
substancial decrease of $|\hat{Z}|$ after the {\em second} parallel decoding step
described above. 

Finally, the output of the decoder gives us a decomposition
$Z=\hat{C}_0+\hat{R}_0+\hat{C}_1+\hat{R}_1$ of the original mismatch
\eqref{eq:mismatch}. For $v\in V_{01}$, the local view $\hat{r}_v$ of $\hat{R}_0$ and the
local view $\hat{c}_v$ of $\hat{C}_1$ give the required local view 
\eqref{eq:evhat} of a candidate global error vector $\hat{\error}$, while for
$v\in V_{10}$ the local view $\hat{e}_v$ of $\hat{\error}$ stems from the local views of
$\hat{R}_0$ and $\hat{C}_1$.

Our main result states that the above decoder will always succeed in producing
an adequate error vector $\hat{\error}$, provided the
initial error weight $|\error|$ is a sufficiently small fraction of the code length.
\begin{theo}\label{thm:main}
There exists a constant $\kappa$, depending only on $\delta$, a lower bound for the
minimum distances of both component codes $C_A$ and $C_B$, such that for large
enough fixed $\Delta$, the above decoding algorithm corrects all error patterns
of weight less than $\kappa n/\Delta^4$ for the quantum Tanner code of length
$n=|Q|$.
\end{theo}

Since the decoder needs only to look at
$Q$-neighbour\-hoods on which the mismatch $|Z|$ is nonzero, by
carefully keeping track of this set of vertices, we obtain a decoder that runs in linear time. 

A direct consequence of the analysis is the soundness of the code, a weaker property than local testability that asserts that sufficiently small errors have a syndrome with a weight proportional to that of the error (local testability would require this to hold for arbitrary errors). This result was previously established for quantum Tanner codes in \cite{HL22}.

Interestingly, our decoder can be adapted to work with the asymptotically good codes of Panteleev and Kalachev~\cite{PK21}. In particular, we show that the theorem above still holds for these codes, provided the error patterns have weight at most $\kappa n/\Delta^6$. Along the way, we point out the hidden relation between the quantum Tanner codes and the lifted product codes of~\cite{PK21}.

\paragraph{Additional comments.}

So as to not let the analysis of the decoder become overly burdensome, we have not tried to make the constant $\kappa$ explicit in
Theorem~\ref{thm:main}, nor have we tried to optimise the dependency in
$\Delta$ of the number of correctable errors. For a similar reason, we have left out possible variations on the decoder. In particular, since one has to resort to
parallel decoding when the sequential decoder is stalled, a natural temptation
is to make the decoder fully parallel and replace the sequential decoding step
by a parallel one. We do not anticipate difficulties of a novel nature in doing
so, but our analysis stems rather naturally from a sequential approach to
decoding, and we have tried to limit the flow of technicalities by avoiding
these variations.

As mentioned in~\cite{PK21} and~\cite{LZ22}, the question of alternatives to random choice for
the component codes $C_A,C_B$ remains open. 
Additionally, finding component codes $C_A,C_B$ with better robustness would simplify our
analysis and yield a fully sequential decoder. On the other hand, showing that
decoding a linear number of adversarial errors is possible with reduced
robustness potentially simplifies the search for constructions of adequate
component codes.

\paragraph{Acknowledgements.}
We would like to thank Benjamin Audoux, Alain Couvreur, Shai Evra, Omar Fawzi,
Tali Kaufman, Jean-Pierre Tillich, and Christophe Vuillot for many fruitful discussions on quantum codes over the years. 
We acknowledge support from the Plan France 2030 through the project NISQ2LSQ, ANR-22-PETQ-0006. GZ also acknowledges support from the ANR through the project QUDATA, ANR-18-CE47-0010.

%%%%%%%%%%%%%%%%%%%%%%%%%%%%
%%%%%%%%%%%%%%%%%%%%%%%%%%%%
%%%%%%%%%%%%%%%%%%%%%%%%%%%%
%%%%%%%%%%%%%%%%%%%%%%%%%%%%
\newpage

\section{Preliminaries}\label{sec:prelim}
\subsection{Expander Graphs}
Let $\G=(V,E)$ be a graph. Graphs will be undirected but may have multiple
edges. 
For $S,T \subset V$, let $E(S,T)$ denote the multiset of edges with one endpoint in
$S$ and one endpoint in $T$.
Let $\G$ be a $\Delta$-regular graph on $n$ vertices, and let
$\Delta=\lambda_1\geq\lambda_2\geq \ldots \geq \lambda_n$ be the eigenvalues of the
adjacency matrix of $\G$. For $n\geq 3$, we define $\lambda(\G):= \max\{|\lambda_i|,
\lambda_i\neq \pm \Delta\}$. 
The graph $\G$ is said to be \emph{Ramanujan} if $\lambda(\G)\leq 2\sqrt{\Delta-1}$.

We recall the following version of the expander mixing lemma (see \textit{e.g.}\
\cite{HLW06}) for bipartite graphs.
\begin{lemma}[Expander mixing lemma] \label{lem:mixing}
Let $\cG$ be a connected $\Delta$-regular bipartite graph on the vertex set $V_0\cup V_1$.
For any pair of sets $S\subset V_0, T \subset V_1$, it holds that
\[ |E(S,T) | \leq \frac{\Delta}{|V_0|} |S| |T| + \lambda(\G) \sqrt{|S| |T|}.\]
\end{lemma}

\subsection{Tanner codes}
\label{subsec:TC}
A binary linear code of length $n$ is an $\f_2$-linear subspace of $\f_2^n$. For
sets $E$ of cardinality $|E|=n$, it will be convenient for us to identify $\f_2^n$
with $\f_2^E$, which we can think of as the space of functions from $E$ to
$\f_2$. Identication with $\f_2^n$ amounts to defining a one-to-one map between $E$ and
$[n]=\{1,2,\ldots ,n\}$, \textit{i.e.}\ a numbering of the elements of $E$.

Let $\G=(V,E)$ be a regular graph of degree $\Delta$, and for any vertex $v$
denote by $E(v)$ the set of edges incident to $v$. 
Assume an identification of $\f_2^{E(v)}$ with $\f_2^\Delta$ for every $v\in V$.
Let $x\in\f_2^{E}$ be a vector indexed by (or a function defined on) the set
$E$. Let us define the {\em local view} of $x$ at vertex $v$ as the subvector
$x_v:=(x_e)_{e\in E(v)}$, \textit{i.e.}\ $x$ restricted to the edge-neighbourhood $E(v)$ of
$v$.

Let $C_0$ be a linear code of length $\Delta$, dimension $k_0=\rho_0\Delta$, and
minimum distance $d_0=\delta_0\Delta$.
We define the Tanner code~\cite{T81} associated to $\G$ and $C_0$ as
\[
\text{Tan}(\G,C_0) :=\{x\in\f_2^E : x_v\in C_0\;\text{for all}\;v\in V\}.
\]
In words, the Tanner code is the set of vectors over $E$ all of whose local
views lie in $C_0$.
By counting the number of linear equations satisfied by the Tanner code, we
obtain
\begin{equation}\label{eq:dimtanner}
\dim \text{Tan}(\G,C_0) \geq (2\rho_0-1)n.
\end{equation}
We also have the bound~\cite{SS96,Gur} on the minimum distance $d$ of the Tanner
code:
\[
d\geq\delta_0 (\delta_0- \lambda(\G)/\Delta)n.
\]
Therefore, if $(\G_i)$ is a family of $\Delta$-regular expander graphs with
$\lambda(\G_i)\leq\lambda <d_0$, and if $\rho_0>1/2$, then the associated family of Tanner
codes has rate and minimum distance which are both $\Omega(n)$, meaning we have
an asymptotically good family of codes, as was first shown in~\cite{SS96}.

\subsection{Quantum CSS codes}

A quantum CSS code is specific instance of a stabilizer code~\cite{got97} that can be defined by two classical codes $\C_0$ and $\C_1$ in the ambient space $\f_2^n$, with the property that $\C_0^\perp \subset \C_1$~\cite{CS96,ste96}. It is a \emph{low-density parity-check} (LDPC) code whenever both $\C_0$ and $\C_1$ are the kernels of sparse parity-check matrices. 
The resulting quantum code $\eQ = (\C_0, \C_1)$ is a subspace of
$(\mathbb{C}_2)^{\otimes n}$, the space of $n$ qubits:
\[ \eQ := \mathrm{Span}\left\{ \sum_{z \in \C_1^\perp} |x+z\rangle \: : \: x\in \C_0 \right\},\]
where $\{ |x\rangle \: : \: x\in \f_2^n\}$ is the canonical basis of $(\mathbb{C}_2)^{\otimes n}$.
The dimension $k$ of the code counts the number of logical qubits and is given by 
\[ k = \text{dim} \, (\C_0/\C_1^\perp) = \text{dim} \, \C_0 + \text{dim} \, \C_1 - n.\]
Its minimum distance is $d = \min (d_X, d_Z)$ with
\[ d_X = \min_{w \in \C_0 \setminus \C_1^{\perp}} |w|, \quad d_Z = \min_{w \in \C_1 \setminus \C_0^\perp} |w|.\]
We denote the resulting code parameters by $\llbracket n,k,d\rrbracket$ and say that a code family $(\Q_n)_n$ is \emph{asymptotically good} if its
parameters are of the form
\[ \llbracket n, k = \Theta(n),d = \Theta(n)\rrbracket.\]

An $n$-qubit Pauli error $E_1 \otimes \ldots \otimes E_n$ with $E_i \in \{ \1, \sigma_X, \sigma_Y, \sigma_Z\}$\footnote{The 1-qubit Pauli matrices are defined by $\1 = \left( \begin{smallmatrix} 1 &0\\0&1 \end{smallmatrix}\right), \sigma_X= \left( \begin{smallmatrix} 0 &1\\1&0 \end{smallmatrix}\right), \sigma_Z= \left( \begin{smallmatrix} 1 &0\\0&-1 \end{smallmatrix}\right)$ and $\sigma_Y = i \sigma_X \sigma_Z$.} is conveniently described by two $n$-bit strings $(e_0, e_1)\in \F_2^n \times \F_2^n$ \textit{via} the mapping 
\[ \1 \mapsto (0,0), \quad \sigma_X \mapsto (1,0), \quad \sigma_Y \mapsto (1,1), \quad \sigma_Z \mapsto (0,1),\]
which forgets global phases. 
The parity-check matrices of $\C_0$ and $\C_1$ give rise to syndrome maps
$\sigma_0, \sigma_1 : \F_2^n \to \F_2^m$ that associate a pair of syndromes
$(\sigma_0(\error_0), \sigma_1(\error_1)) \in \F_2^m \times \F_2^m$ to any
$n$-qubit Pauli error $(\error_0, \error_1)\in \F_2^n \times \F_2^n$.
The decoding problem for a stabilizer code is as follows: given a syndrome
$(\sigma_0(\error_0), \sigma_1(\error_1))$, recover the error up to an element
of the stabilizer group, that is return $(\hat{\error}_0, \hat{\error}_1)$ such
that $\error_0 + \hat{\error}_0 \in \C_1^\perp$ and $\error_1 + \hat{\error}_1 \in \C_0^\perp$. 

While an optimal decoding of \emph{random} errors would typically exploit
possible correlations between $\error_0$ and $\error_1$, it is always possible
to correct both errors independently. Here, we will be concerned with the
adversarial setting where $\error_0$ and $\error_1$ are of sufficiently low
weight, but otherwise arbitrary. In that case, both errors should be decoded
independently, and we will focus on the case where $(\error_0=0, \error_1=e)$ in the following.

\subsection{Tensor codes and dual tensor codes: robustness}
Definitions and results for this section are taken from~\cite{LZ22}, to which we
refer for proofs and comments.

Let $A$ and $B$ be two sets of size $\Delta$. 
We define codes on the ambient space $\f_2^{A\times B}$
that we may think of as the space of matrices whose rows (columns) are indexed by
$A$ (by $B$). If $C_A\subset \f_2^A$ and $C_B\subset\f_2^B$ are two linear codes,
we define the {\em tensor} (or product) code $C_A\otimes C_B$ as the space of
matrices $x$ such that for every $b\in B$ the column vector $(x_{ab})_{a\in A}$
belongs to $C_A$ and for every $a\in A$ the row vector $(x_{ab})_{b\in B}$
belongs to $C_B$. It is well known that $\dim(C_A\otimes
C_B)=\dim(C_A)\dim(C_B)$ and that the minimum distance of the tensor code is $d(C_A\otimes C_B)=d(C_A)d(C_B)$.

Consider the codes $C_A\otimes\f_2^B$ and $\f_2^A\otimes C_B$
consisting respectively of the space of matrices whose columns are codewords of
$C_A$ and whose rows are codewords of $C_B$. We may consider their sum
$C_A\otimes\f_2^B + \f_2^A\otimes C_B$ which is called a {\em dual
tensor} code, since it is the dual code of the tensor code
$C_A^\perp\otimes C_B^\perp=(C_A^\perp\otimes\f_2^B)\cap(\f_2^A\otimes
C_B^\perp)$. It is relatively straightforward to check that $d(C_A\otimes\f_2^B +
\f_2^A\otimes C_B)=\min(d(C_A),d(C_B))$.

\begin{definition}\label{def:robust}
Let $0\leq w\leq\Delta^2$. Let $C_A$ and $C_B$ be codes of length $\Delta$ with minimum distances $d_A$ and
$d_B$. We shall say that the dual tensor code $C=C_A\otimes\f_2^B +
\f_2^A\otimes C_B$ is $w$-{\em robust}, if for any codeword $x\in C$
of Hamming weight $|x|\leq w$, there exist $A'\subset A, B'\subset B$, $|A'|\leq
|x|/d_B$, $|B'|\leq |x|/d_A$, such that $x_{ab}=0$ whenever $a\notin A',
b\notin B'$.
\end{definition}

\begin{proposition}
Let $C_A$ and $C_B$ be codes of length $\Delta$ with minimum distances $d_A$ and
$d_B$, and suppose $C=C_A\otimes\f_2^B +
\f_2^A\otimes C_B$ is $w$-robust with $0<w<d_Ad_B$. Then for any codeword $x\in
C$ such that $|x|\leq w$,  there exist $A'\subset A, B'\subset B$, $|A'|\leq
|x|/d_B$, $|B'|\leq |x|/d_A$ and 
a decomposition $x = c + r$, with $c \in C_{A} \otimes \F_2^{B'}$ and $r \in \F_2^{A'}
\otimes C_B$. 
\end{proposition}

\begin{proof}
To see this, apply the definition and write $x=r'+c'$, with
$r'_{ab}=c'_{ab}$ for any $(a,b)\in (A\setminus A')\times (B\setminus B')$.
The restrictions of $r'$ and $c'$ to $(A\setminus A')\times
(B\setminus B')$ both belong to the code obtained by tensoring $C_A'$ and
$C_B'$, the punctured codes deduced from $C_A$ and $C_B$ by throwing away
coordinates of $A'$ and $B'$. This code is the same as the punctured code
obtained from 
$C_A\otimes C_B$ by throwing away the coordinates $A'\times B \cup
A\times B'$. Therefore, there exists a tensor codeword of $C_A\otimes C_B=C_{A}
\otimes \F_2^{B}\cap \F_2^A\otimes C_B$ that
coincides with $c'=r'$ on $(A\setminus A')\times
(B\setminus B')$: adding this tensor codeword to both $c'$ and $r'$ yields the
required pair $r,c$ such that $x=r+c$.
\end{proof}

\begin{prop}\label{prop:tensor}
Let $C_A$ and $C_B$ be codes of length $\Delta$ and minimum distances $d_A, d_B$
such that the dual tensor code $C_A\otimes\f_2^B + \f_2^A \otimes C_B$ is
$w$-robust with $w\leq d_Ad_B/2$. Then, any word $x$ close to both the column
and row code is also close to the tensor code: precisely,
if $d(x, C_A \otimes \f_2^{B})+ d(x, \f_2^{A} \otimes C_B) \leq w$ then 
\begin{align}\label{eqn:robust}
d(x, C_A \otimes C_B) \leq \frac{3}{2} \left( d(x, C_A \otimes \f_2^{B})+ d(x, \f_2^{A} \otimes C_B) \right).
\end{align}
\end{prop}

\begin{definition}\label{def:resistance}
Let $C_A\subset \f_2^A$ and $C_B\subset \f_2^B$. 
For integers $w,p$,
let us say that the dual tensor code $C_{A}\otimes\f_2^{B}+\f_2^{A}\otimes
C_{B}$ is $w$-robust with $p$-resistance to puncturing, 
if
for any  $A'\subset A$ and $B'\subset B$ such that $|A'|=|B'|=\Delta-w'$, with
$w'\leq p$, the dual tensor code
\[
C_{A'}\otimes\f_2^{B'}+\f_2^{A'}\otimes C_{B'}
\]
is $w$-robust.
\end{definition}

We shall need the following result on the robustness of random dual tensor
codes. 

\begin{restatable}{theo}{Probust}
\label{thm:P-robust}
Let $0<\rho_A < 1$ and $0<\rho_B<1$. Let $0<\varepsilon<1/2$ and $1/2+\eps<\gamma <1$. Let $C_A$ be a random code obtained from a
random uniform $\rho_A\Delta\times\Delta$ generator matrix, and let $C_B$ be a
random code obtained from a random uniform $(1-\rho_B)\Delta\times\Delta$
parity-check matrix. 
With probability tending to~$1$ when $\Delta$ goes to
infinity, the dual tensor code
\[
C_{A}\otimes\f_2^{B}+\f_2^{A}\otimes C_{B}
\]
is $\Delta^{3/2-\varepsilon}$-robust with $\Delta^\gamma$-resistance to puncturing.
\end{restatable}

%%%%%%%%%%%%%%%%%%%%%%%%%%%%
%%%%%%%%%%%%%%%%%%%%%%%%%%%%
%%%%%%%%%%%%%%%%%%%%%%%%%%%%
\section{Detailed description of quantum Tanner codes}
\label{sec:details}

We recall the ingredients that enable us to define the family of asymptotically good quantum
LDPC codes introduced in~\cite{LZ22}. 
The only difference with~\cite{LZ22} is that we prefer to work here with a quadripartite version of the Cayley complex since it gives a simpler description of the decoding algorithm. 
This quadripartite version was
mentioned in passing in~\cite{LZ22}, 
and is not much more than a particular instance of the bipartite version,
and is also suited to the lower bound of Theorem~\ref{thm:good-qLDPC}
 on the minimum distance.
Below we discuss the quadripartite structure in some more detail.

\subsection{Left-right Cayley complexes (quadripartite version)}
\label{subsec:LRCayley}

The square complex we shall rely on for the construction first appeared in \cite{PK21} as a balanced product of double covers of non-bipartite Cayley graphs. For the sake of simplicity, we will rather use the language of left-right Cayley complexes in their quadripartite version.
A {\em left-right Cayley complex} $X$ is introduced in~\cite{DEL21} from a group
$G$ and two sets of generators $A=A^{-1}$ and $B=B^{-1}$. As in~\cite{DEL21} we
will restrict ourselves, for the sake of simplicity, to the case
$|A|=|B|=\Delta$. The complex is made up of vertices, $A$-edges, $B$-edges, and
squares. The vertex set consists of four copies of the group $G$ in the
quadripartite version, $V = V_{00} \cup V_{10} \cup V_{01}\cup V_{11}$ with
$V_{ij} = G \times \{ij\}$. 
The advantage of this quadripartite version, also considered in \cite{PK21} and \cite{gol21}, is that it does not require any additional assumption on the choice of group and generators, for instance that $ag \ne gb$ for all $g\in G, a\in A, b \in B$, as in \cite{DEL21}.
We will also use the notation $V_0:=V_{00}\cup
V_{11}$ and $V_1:=V_{01}\cup V_{10}$.
The $A$-edges are pairs of vertices of the form $\{(g,i0) ,(ag,i1)\}$ and $B$-edges are of
the form $\{(g,0j),(gb,1j)\}$ for $g\in
G,a\in A,b\in B$, $i,j=0,1$. We denote by $E_A$ and $E_B$ these two edge sets. The associated graphs are denoted by $\G_A = (V, E_A)$ and $\G_B = (V,E_B)$.
 A {\em square} is a set of four vertices of the form $\{(g,00),(ag,01),(gb,10),(agb,11)\}$.
The set of squares (or quadrangles) of the complex is denoted by $Q$.
Every vertex is incident to exactly $\Delta^2$ squares. For a vertex $v$, the set of incident squares is called the $Q$-neighbourhood,
and denoted by $Q(v)$. 

The sets of generators $A$ and $B$ will be chosen so that the Cayley graphs
$\Cay(G,A)$ and $\Cay(G,B)$ are non-bipartite Ramanujan graphs.
It should be understood that when writing $\Cay(G,A)$ we implicitely mean the
Cayley graph defined by left multiplication by
elements of $A$, while $\Cay(G,B)$ stands for the Cayley graph defined by
right multiplication by elements of $B$. The sets $A$ and $B$ could in principle
be chosen to be identical, but we keep a distinct notation for both sets, in
particular in order to allow the above abuse of notation to be non-confusing.

We see that the subset of edges of $E_A$ that connect vertices of $V_{00}$ to
vertices of $V_{01}$ make up a double cover of the Cayley graph $\Cay(G,A)$,
the edges $E_A$ that connect $V_{10}$ to $V_{11}$ make up a second copy of the
same double cover. Therefore, the graph $\G_A$ is a disjoint union of two copies
of the double cover of $\Cay(G,A)$. Similarly, $\G_B$ is a disjoint union of two
copies of the double cover of $\Cay(G,B)$. We will regularly talk about
expansion in $\G_A$ (or $\G_B$) to mean expansion in either of the connected
components of $\G_A$ ($\G_B$).

Let us introduce one additional graph that exists on the complex $X$,
and that we denote by $\G^\square$. This graph
puts an edge between all pairs of vertices of the form $\{(g,i),(agb,i)\}$, $g\in G,a\in A,b\in
B, i=0,1$. The graph $\G^\square$ is therefore made up of two connected
components, on $V_0$ and $V_1$, that we denote by $\G_0^\square$ and
$\G_1^\square$. We note that $\G^\square$ is regular of degree $\Delta^2$, and
may have multiple edges. 

If $\Cay(G,A)$ and $\Cay(G,B)$ are Ramanujan, then $\G^\square$ inherits some
of their expansion properties. Specifically:

\begin{lemma}\label{lem:lambda}
Assume that $\Cay(G,A)$ and $\Cay(G,B)$ are Ramanujan graphs, then 
\[ 
\lambda(\G_0^\square) \leq 4\Delta, \quad \lambda(\G_1^\square)\leq 4\Delta.\]
\end{lemma}

The proof follows from the fact that the adjacency matrix of $\G^\square$ is the
product of the adjacency matrices of $\G_A$ and $\G_B$, and that these two
adjacency matrices commute, by definition of the square complex. See~\cite{LZ22}
for a little more detail.

\subsection{Labelling $Q$-neighbourhoods}
We will define Tanner codes on $\G_0^\square$ and $\G_1^\square$, which implies a labelling of
the coordinates in every $Q$-neighbourhood $Q(v)$. There is a natural labeling
of $Q(v)$ by the set $A\times B$, namely a one-to-one map $\phi_v~: A\times B
\to Q(v)$,  which we now state explicitely.

We set
\begin{align*}
\text{for}\; v=(g,00)\in V_{00},\quad& \phi_v(a,b) =
\{(g,00),(ag,01),(gb,10),(agb,11)\},\\
\text{for}\; v=(g,01)\in V_{01},\quad& \phi_v(a,b) =
\{(g,01),(a^{-1}g,00),(gb,11),(a^{-1}gb,10)\},\\
\text{for}\; v=(g,10)\in V_{10},\quad& \phi_v(a,b) =
\{(g,10),(ag,11),(gb^{-1},00),(agb^{-1},01)\},\\
\text{for}\; v=(g,11)\in V_{11},\quad& \phi_v(a,b) =
\{(g,11),(a^{-1}g,10),(gb^{-1},01),(a^{-1}gb^{-1},00)\}.
\end{align*}
The map $\phi_v$ thus defined is obviously one-to-one, and one easily checks
that:

\medskip

{\em
Any two vertices $v=(g,i0)$ and $w=(gb,i1)$, $i=0,1$, that are connected through
a $B$-edge (labelled $b$), have a common ``column'', \textit{i.e.}\ their
$Q$-neighbourhoods share exactly $\Delta$ squares that are labelled $(a,b), a\in
A$, in both $Q(v)$ and $Q(w)$.
}
\medskip

Similarly,

\medskip

{\em
Any two vertices $v=(g,0i)$ and $w=(ag,1i)$, $i=0,1$, that are connected through
an $A$-edge (labelled $a$), have a common row, \textit{i.e.}\ their
$Q$-neighbourhoods share exactly $\Delta$ squares that are labelled $(a,b), b\in
B$, in both $Q(v)$ and $Q(w)$.
}

\medskip

The situation is illustrated on Figure~\ref{fig:code}. Summarising, any two
vertices connected by $B$-edge (an $A$-edge) have a common column (row) in their $Q$-neighbourhoods,
that is labelled by the same $b\in B$ ($a\in A$). This is the reason for the
possibly intriguing inversions in the definition of $\phi_v$: without these
inversions the $Q$-neighbourhoods of two neighbouring vertices would still share
a common row or a common column, but their indexes in their respective local
views would be inverse of each other. This would still be manageable but
slightly less convenient.

\subsection{Local codes}

The constraints of a Tanner code consist of local constraints from small codes enforced on the local view of each vertex. 
For quantum Tanner codes, now that all local $Q$-neighbourhoods are isomorphic
to $A \times B$, we may put local constraints that are codewords of the tensor
codes $C_A \otimes C_B$ and $C_A^\perp \otimes C_B^\perp$.

Recall that the generators of the quantum Tanner code correspond to a basis of $C_A \otimes C_B$ in each local view of $V_{00} \cup V_{11}$ (for the $\sigma_Z$-type generators) and to a basis of $C_A^\perp \otimes C_B^\perp$ in each local view of $V_{01} \cup V_{10}$ (for the $\sigma_X$-type generators). 
The classical code $\C_0\subset \F_2^Q$ correcting $\sigma_X$-type errors is the Tanner code on the graph $\G_0^\square$ with local constraints corresponding to the dual tensor code $(C_A \otimes C_B)^\perp = C_A^\perp \otimes \F_2^B + \F_2^A \otimes C_B^\perp$.
With the notation of Section~\ref{subsec:TC}, $\C_0=\text{Tan}(\G_0^\square,
C_A^\perp \otimes \F_2^B + \F_2^A \otimes C_B^\perp)$.
Similarly, the classical code $\C_1\subset\F_2^Q$ correcting $\sigma_Z$-type errors is the Tanner
code on the graph $\G_1^\square$ with local constraints corresponding to the
dual tensor code $(C_A^\perp \otimes C_B^{\perp})^\perp = C_A \otimes \F_2^B +
\F_2^A \otimes C_B$, \textit{i.e.}\ $\C_1=\text{Tan}(\G_1^\square, C_A \otimes \F_2^B +
\F_2^A \otimes C_B)$.

\paragraph{Summary.}
A large enough $\Delta$ is chosen, together with an infinite family of groups
$G$ with generating sets $A,B$, $|A|=|B|=\Delta$, such that the left Cayley
graph $\Cay(G,A)$ and the right Cayley graph $\Cay(G,B)$ are Ramanujan. The
quadripartite left-right
square complex $X$ is defined by $G,A,B$.

For the conditions for the component codes that together with the above square
complexes $X$ will yield asymptotically good quantum codes, we recall Theorem~16 from~\cite{LZ22}.

\begin{theo}\label{thm:123}
Fix $\rho\in(0,1/2)$, $\eps \in (0,1/2)$, $\gamma\in (1/2+\eps,1)$ and $\delta>0$. 
If $\Delta$ is large enough and $C_A$ and $C_B$ are codes of length $\Delta$
such that
\begin{enumerate}
\item $0<\dim C_A \leq\rho\Delta$ and $\dim C_B=\Delta-\dim C_A$,
\item the minimum distances of $C_A,C_B,C_A^\perp,C_B^\perp$ are all
$\geq\delta\Delta$,
\item \label{cond3} both dual tensor codes $C_0^\perp=(C_A\otimes C_B)^\perp$ and
$C_1^\perp=(C_A^\perp\otimes C_B^\perp)^\perp$ are $\Delta^{3/2-\eps}$-robust
with $\Delta^\gamma$-resistance to puncturing (see
Definition~\ref{def:resistance}),
\end{enumerate}
then the quantum code $\eQ=(\C_0,\C_1)$ has length $|Q|$, dimension at least $(1-2\rho)^2|Q|$
and minimum distance at least $|Q|\delta/4\Delta^{3/2+\eps}$.
\end{theo}

It was proved in~\cite{LZ22} (see also Theorem~\ref{thm:P-robust}), that for any fixed $\rho$ and $\delta >0$ such that
$-\delta\log_2\delta-(1-\delta)\log_2(1-\delta)<\rho$, then for any fixed
$\eps,\gamma$ and $\Delta$ large enough, randomly choosing $C_A$ and $C_B$ with
the required rates will yield codes that satisfy conditions 2 and 3 in
Theorem~\ref{thm:123}  with high probability. In the sequel we will set,
somewhat arbitrarily, $\gamma=1-\eps$, so as to minimise the number of
constants, and we will naturally assume that the codes $C_A$ and $C_B$ satisfy
the conditions of Theorem~\ref{thm:123}.

%%%%%%%%%%%%%%%%%%%%%%%%%%%%%%%%%%

\section{The decoding strategy}
\label{sec:decoder}

We recall that we consider without loss of generality a bit-flip error $\error \in \F_2^\Q$
and wish to correct it with the help of the classical Tanner code $\C_1$. The
goal of the decoder is to output some guess $\hat{\error} \in \F_2^\Q$ and it is
successful if $\error + \hat{\error} \in \C_0^\perp$.

For a vertex $v\in V$, denote by $e_v$ the local view of $\error$ on $Q(v)$,
i.e. its restriction to $Q(v)$, which we also extend back to $\F_2^Q$ by padding
it with $0$s.
The error $\error$ can be identified with the collection of local views $\error
= \{e_v\}_{v \in V_1}$ and we note that they satisfy $\sum_{v \in V_{01}} e_v =
\sum_{v \in V_{10}} e_v$, since $(Q(v))_{v\in V_{01}}$ and $(Q(v))_{v\in
V_{10}}$ are both partitions of $Q$. The Hamming weight of $\error$ is
\[ |\error| = \sum_{v \in V_{01}} |e_v| = \sum_{v \in V_{10}} |e_v|.\]
As observed in~\cite{LZ22}, a possible approach to decoding a Tanner code is to
consider the \emph{mismatch} of the error $\error$. It is defined as follows.
For each vertex $v \in V_1$, one can compute an error $\eps_v$ of minimal Hamming weight yielding the corresponding local syndrome. This gives a decomposition of the local views of the error $e_v = \eps_v + c_v + r_v$ with $c_v \in C_A\otimes \F_2^B, r_v = \F_2^A \otimes C_B$, and $\eps_v$ of minimal Hamming weight. 
In the case where $e_v = \eps_v$ for all $v\in V_1$, then the decoder has succeeded in recovering the true error. In general, however, we have $e_v \ne \eps_v$ for some $v \in V_1$. 
The mismatch vector in $\F_2^Q$ defined as $Z := \sum_{v \in V_1} \eps_v$ then characterises the inconsistency between the local views in $V_{10}$ and $V_{01}$.  From $\sum_{v \in V_1} e_v = 0$, we obtain that
\[ Z = \sum_{v \in V_1} r_v +c_v.\]
We observe that the minimality of $\eps_v$ implies that $|e_v + c_v + r_v| \leq |e_v|$ and therefore $|c_v + r_v| \leq 2|e_v|$. This immediately shows that 
\[ |Z| \leq \sum_{v \in V_1} |r_v + c_v| \leq 2 \sum_{v \in V_1} |e_v| = 4
|\error|.\]
We note that the quantum Tanner codes are \textit{a priori} not locally testable, and therefore the Hamming weight of the mismatch can be much smaller than that of the error (otherwise a simple test for detecting an error would be to sample bits of the mismatch). 

The idea behind the decoder will be to find a decomposition $\{\hat{r}_v, \hat{c}_v\}_{v \in V_1}$ such that $Z =  \sum_{v \in V_1} \hat{r}_v +\hat{c}_v$.
In that case, the decoder will return the following guess for the decomposition of the error: $\{\hat{e}_v\}_{v \in V_1}$ with 
\[ \hat{e}_v := \eps_v + \hat{r}_v + \hat{c}_v.\]
In particular, the vectors $\error$ and $\hat{\error}$ differ by an element of $\C_1$ since they both give the same syndrome, and
a sufficient condition to guarantee the success of the decoder is that $|\error
+ \hat{\error}| < d_{\min}(\eQ)$. 

Since each column $A\times\{b\}$ of a $Q$-neighbourhood $Q(v)$  appears also in the
$Q$-neighbourhood of a neighbouring vertex, we have that for either $v\in
V_{00}$ or $v\in V_{10}$, 
each column subvector of $c_v$ is a codeword of $C_A$ that appears in two local
views, one indexed by a vertex of $V_{10}$ and one by a vertex of $V_{00}$:
 for $v\in V_{01}\cup V_{11}$,
we have that any column subvector of $c_v$ lies both in the local view around a
vertex of $V_{01}$ and around a vertex of $V_{11}$.

A similar observation also holds for row codewords:
This implies that one can also define couples $(c_v, r_v)$ for all vertices of $V_0$. 

It is convenient to define four vectors associated with the collection $\{r_v, c_v\}_{v \in V_1}$,
\[ C_0 = \sum_{v \in V_{10}} c_v, \quad R_0 = \sum_{v \in V_{01}} r_v, \quad C_1 = \sum_{v \in V_{01}} c_v, \quad R_1 = \sum_{v \in V_{10}} r_v,\]
that provide a decomposition of the mismatch:
\[ Z = C_0 + R_0 + C_1+R_1.\]
Note that $(C_0, R_0, C_1, R_1)$ and $\{r_v,c_v\}_{v\in V}$ are in one-to-one correspondence.

We will denote this decomposition as
\[ \mathcal{Z} = (C_0, R_0, C_1, R_1) = \{r_v,c_v\}_{v\in V}.\]
The decoder then simply aims at finding a valid decomposition $\hat{\mathcal{Z}} = (\hat{C}_0, \hat{R}_0,\hat{C}_1,\hat{R}_1)$ of $Z$  and this decomposition will yield the correct guess provided that $|\error+\hat{\error}| = |\sum_{v \in V_{10}} c_v + r_v + \hat{c}_v+\hat{r}_v| = |C_0+R_1+\hat{C}_0+\hat{R}_1|$ is less than the minimal distance of the quantum Tanner code.

A natural idea for the decoder is to progressively decrease the weight of $Z$,
so as to obtain a consistent assignment that will correspond to the output of
the decoder. To do this, one can look for some vertex $v$, as well as some local
codeword $c_v + r_v$ of the dual tensor code such that $|Z + c_v + r_v| < |Z|$.
If it was always possible to guarantee the existence of such a vertex and local
codeword, then one would simply repeat the operation until reaching a complete
decomposition of the mismatch. At each step, the weight of the mismatch
decreases by at least 1, so at most $|Z|$ steps are required. In addition, we
observe that $|\hat{C}_0+\hat{R}_1| \leq \Delta^2 |Z| \leq 4 \Delta^2 |\error|$ and
we recall that $|C_0+R_1| = \sum_{v \in V_{10}} |c_v + r_v| \leq 2 |\error|$. 
We conclude that if we could always find a vertex and a local codeword that could be added to the mismatch to decrease its weight, then the decoder would return a valid correction provided the initial error is not too large:
\[ |\error| \leq \frac{d_{\min}(\eQ)}{4 \Delta^2 +2}.\]
Unfortunately, we only know how to prove the existence of local codewords decreasing the weight of $Z$ if the local dual tensor code is $w$-robust with $w >\Delta^{3/2}$, and we recall that Theorem \ref{thm:P-robust} only asserts that randomly chosen codes $C_A \otimes \F_2^B + \F_2^A \otimes C_B$ and $C_A^\perp \otimes \F_2^B + \F_2^A \otimes C_B^\perp$ are $w$-robust for $w = O(\Delta^{3/2-\eps})$ for random codes $C_A, C_B$.

For this reason, we need to tweak this natural decoder in order to deal with
error configurations (described by their mismatch $Z$) whose weight cannot be
decreased locally by the addition of some $c_v +r_v$. For this purpose, we will
add a \emph{parallel} procedure that will be described in detail in Section~\ref{sec:analysis}.
The actual decoder will apply the sequential local procedure described
above, though we will only need to apply it to vertices of $V_0$, and resort to
the parallel procedure whenever the sequential decoder is stalled.
The parallel procedure will consist of two phases, the first will be applied to vertices of
$V_1$, and the second to vertices of $V_0$.
While it may not be possible to decrease locally the Hamming weight of the
mismatch when the sequential decoder is stalled, we will show that together, the
two steps of the parallel procedure will significantly decrease the Hamming weight
of $Z$, by a factor which becomes close to $\sqrt{\Delta}$ when the exponent
$\eps$ that defines the robustness parameter $w=\Delta^{3/2-\eps}$ tends to $0$. 
Since the sequential procedure decreases $|Z|$ at every step, the number of sequential steps is at most linear in
$|Z|$, and therefore in the Hamming weight $|\error|$ of the initial error,
provided it is not too large. Since any parallel step divides $|Z|$ by a
constant,  the total number of parallel steps is logarithmic in $|\error|$. 

In both cases, sequential and parallel, the decoder needs only look at
$Q$-neighbour\-hoods on which the mismatch $|Z|$ is nonzero: therefore, by
carefully keeping track of the set of vertices for which the mismatch is
nonzero, the decoder can run in linear-time in the standard uniform cost
computational model.

\paragraph{Norm and minimal decomposition of the mismatch.}
The analysis of the decoder will rest on a careful study of a situation where the sequential decoder is stalled, that is, where there does not exist any vertex (in $V_0$) where one can add a codeword of the dual tensor code that would decrease the Hamming weight of $Z$. (If the sequential decoder is never stalled, then there is essentially nothing to prove.)
The parallel procedure that will unlock the situation will consist of two steps. While the combination of both steps will reduce the Hamming weight $|Z|$ of the mismatch, it is not the case of the first part alone. In order to keep track of the progress of the decoding algorithm, we will therefore require another metric beside the Hamming weight: this role will be played by the norm of the decomposition, that will be a proxy for the number of vertices of $V_0$ for which $c_v + r_v \ne 0$.

First, note that there are many valid decompositions of the mismatch since adding any codeword of the tensor code $C_A \otimes C_B$ to both $r_v$ and $c_v$ does not change their sum. 
 For an element $r \in \F_2^A \otimes C_B$ or $c \in C_A \otimes \F_2^B$, we denote by $\|r\|$ or $\|c\|$ the number of nonzero rows (for $r$) or columns (for $c$) in its support:
\begin{align*} 
\|r\| &:= \min \{|A'| \; \text{s.t.} \; A' \subseteq A, \mathrm{supp}(r) \subset \F_2^{A'} \otimes C_B\}, \\
\|c\| &:= \min \{|B'| \; \text{s.t.} \; B' \subseteq B, \mathrm{supp}(c) \subset  C_A \otimes \F_2^{B'}\}.
\end{align*}
The norm $\|\mathcal{Z}\|$ of the decomposition is then simply:
\begin{align*}
\| \mathcal{Z}\| &:= \|C_0\| + \|R_0\| + \|C_1\| + \|R_1\|,
\end{align*}
with 
\[\|C_0\| := \sum_{v \in V_{10}} \|c_v\|, \quad \|R_0\| := \sum_{v \in V_{01}} \|r_v\|, \quad  \|C_1\| := \sum_{v \in V_{01}} \|c_v\|, \quad \|R_1\|: = \sum_{v \in V_{10}} \|r_v\|.\]

It will be useful for us to consider a \emph{minimal} decomposition of the mismatch, $\mathcal{Z} = (C_0, R_0, C_1, R_1)$, which simply corresponds to a decomposition of minimal norm.
We define the norm $\|Z\|$ of the mismatch to be the norm of a such minimal decomposition:
\[ \|Z\| := \min_{\text{$\mathcal{Z}$ is a decomposition of $Z$}} \|\mathcal{Z}\|.\]

A minimal decomposition has the property that for each vertex, it is not possible to add a codeword of $C_A \otimes C_B$ that decreases the norm of the local view: for each $v \in V$, and each $ u \in C_A \otimes C_B$, it holds that 
\[ \| r_v + c_v + u\| \geq \|r_v+c_v\|.\]

%%%%%%%%%%%%%%%%%%%%%%%%%%%%
%%%%%%%%%%%%%%%%%%%%%%%%%%%%
%%%%%%%%%%%%%%%%%%%%%%%%%%%%
%%%%%%%%%%%%%%%%%%%%%%%%%%%%

\section{Analysis of the decoder}
\label{sec:analysis}

\subsection{The decoding algorithm}\label{sec:tweaked}

The goal of the algorithm is to find a small decomposition of $Z$. 
We will keep track of 5 vectors initialized as follows
\[ \hat{Z} := Z, \quad \hat{C}_0 := 0, \quad \hat{R}_0:=0, \quad \hat{C}_1 :=0,
\quad \hat{R}_1 :=0.\]
It will alternate between two procedures:
\paragraph{Sequential procedure.} 
This procedure is relatively natural, and
consists of modifying $\hat{Z}$ locally, \textit{i.e.}\ in some $Q(v)$, so as to decrease the
Hamming weight of $\hat{Z}$. Specifically,
while there exists some $v\in V_{00}$ or $v\in V_{11}$, and $c_v \in C_A \otimes \F_2^B, r_v \in  \F_2^A \otimes C_B$ such that $|\hat{Z} + c_v +r_v| < |\hat{Z}|$, perform the update:
\[ \hat{Z} \leftarrow \hat{Z} + c_v +r_v, \quad \hat{C}_j \leftarrow \hat{C}_j
+c_v, \quad \hat{R}_i \leftarrow  \hat{R}_i +r_v.\]

\paragraph{Parallel procedure.}
Ideally, we would like the sequential procedure to be applicable until we have
$\hat{Z}=0$, at which point decoding would be complete. However, we cannot
always guarantee
the existence of a vertex for which we can decrease the Hamming weight of $\hat{Z}$. The sequential decoder may
be stalled. 
At this point we switch to an alternative procedure which will
unlock the situation. We remark that we could try to extend the sequential
procedure to vertices of $V_1$ and not just apply it to vertices of $V_0$ as
specified above. However, we cannot guarantee that this extended sequential
decoder will not stall as well, and we will deal with the stalled decoder
through a procedure whose analysis only needs to know that
sequential decoder is stalled on $V_0$.
This alternative procedure consists of two {\em parallel}
decoding steps.

The {\em first parallel decoding step} is the most involved. It consists of looking
for all vertices $v$  of $V_1$ for which the weight of $\hat{Z}$ can be decreased,
not on
{\em all} their $Q$-neighbourhood $Q(v)$, but on a {\em sufficiently large
subset} of
$Q(v)$, \textit{i.e.}\ indexed by some
$A_0\times B_0$ for $A_0\subset A$, $B_0\subset B$. We give below the precise
criteria for vertices $v$ for which the decoder updates $\hat{Z}$ on $Q(v)$.
We cannot guarantee that this first parallel decoding step decreases the Hamming
weight of $\hat{Z}$. However, we will show that this first parallel decoding step 
significantly decreases {\em the number of active vertices of $V_0$}: the 
active vertices $v$ of $\hat{Z}$, for $v\in V_{ij}$, are the vertices on whose $Q$-neighbourhoods 
we have $R_i+C_j\neq 0$, for a relevant decomposition
$\mathcal{Z}=(C_0,R_0,C_1,R_1)$ of $\hat{Z}$. We shall be more precise with the
definition of this set when we start the analysis.
Specifically
it divides this number by a quantity close to $\Delta$. Proving this (as stated in Theorem~\ref{thm:tweaked}) will
be the most technical part of the analysis of the decoder.

The \emph{second parallel decoding step} is conceptually simpler. It consists simply of
looking for all vertices of $V_0$ on which the value of $|\hat{Z}|$ can be decreased 
and \emph{simultaneously} updating $\hat{Z}$ on all such vertices.
Precisely: 

\paragraph{Second parallel decoding step.}
Simultaneously, for each $v \in V_0$, search for codewords $c_v \in C_A \otimes \F_2^B, r_v \in  \F_2^A \otimes C_B$ 
such that $|\hat{Z} + c_v+r_v|<|\hat{Z}|$. If such $c_v,r_v$ exist, choose $c_v,r_v$
that maximize the difference $|\hat{Z}|-|\hat{Z} + c_v+r_v|$ and perform the update
\[ \hat{Z} \leftarrow \hat{Z} + c_v +r_v, \quad \hat{C}_j \leftarrow \hat{C}_j
+c_v, \quad \hat{R}_i\leftarrow \hat{R}_i+r_v\]
for the appropriate $i,j \in \{0,1\}$.

Just by relying on the fact that the number of active vertices of $V_0$ has been
greatly reduced after the first parallel decoding step, it will be relatively
straightforward to show that afer the second parallel decoding step, the Hamming weight
of $\hat{Z}$ must incur a significant reduction from its value from {\em before the first
parallel decoding step.} Together, both parallel steps, first on $V_1$ and then on
$V_0$, enable us to unlock the situation and divide the Hamming weight of $\hat{Z}$ by a
quantity that we will show to be close to $\sqrt{\Delta}$.

We now finish this description of the decoder by giving the precise description
of the first parallel decoding step. We need to recall that the dual tensor code
$\DT{A}{B}$ is $w$-robust with resistance to puncturing $p$ with
$w=\Delta^{3/2-\eps}$ and $p\geq\Delta^{\gamma}$, where $\eps >0$ can be taken
arbitrarily close to $0$ and $\gamma <1$ can be taken arbitrarily close to $1$.
For simplicity we will choose $\gamma = 1-\eps$.

\paragraph{First parallel decoding step.}
The {\em tweaked} decoder does the following: on $V_{1}$ it looks, {\em in parallel}, for vertices
$v$ for which it can identify a partition $A=A_0\cup A''$ 
and a partition $B=B_0\cup B''$ such that
$|A''|\leq\Delta^\gamma/2$ and $|B''|\leq\Delta^\gamma/2$ and which satisfy the
following properties:
\begin{itemize}
\item[--] the value of $|\hat{Z}|$ on the restricted local view $A_0\times B_0$ of $Q(v)$ is
$>w/2$,
\item[--]
there exists $r_v+c_v\in\DT{A_0}{B_0}$ 
such that $|r_v+c_v|>w$, and replacing
$\hat{Z}$ by $\hat{Z}+r_v+c_v$ decreases the value of $|\hat{Z}|$ on $A_0\times B_0$ so that it becomes
$<w/2$. 
\end{itemize}
The behaviour of the decoder on the $Q$-neighbourhood of such a vertex $v$ is then to
\begin{itemize}
\item[--]  first update $\hat{Z}$ to $\hat{Z}+r_v+c_v$, update $\hat{R}_i$ to
$\hat{R}_i+r_v$, and  $\hat{C}_j$ to $\hat{R}_j+c_v$,
\item[--] look for individual rows and columns on $Q(v)$, on which it is possible to
decrease the weight of $\hat{Z}$ by adding a single (column) codeword of $C_A$ or a
single (row) codeword of $C_B$. If such rows or columns exist it updates
$\hat{Z}, \hat{R}_i, \hat{C}_j$
correspondingly, and stops its treatment of $Q(v)$ when there are no more such
rows or columns. 
\end{itemize}
We stress that this whole procedure is applied \emph{in parallel} to all
$Q$-neighbourhoods from the original stalled $\hat{Z}$.

\subsection{Analysis of the parallel procedure}

We suppose the sequential decoder is stalled. We let $\mathcal{Z} = (R_0, C_0, R_1, C_1)$
be a minimum representation 
of the current stalled mismatch
$\hat{Z}$, meaning that $\|\mathcal{Z}\| = \|\hat{Z}\|$.
In particular
we will keep in mind that $\mathcal{Z}$ is {\em locally minimal}, meaning that its norm cannot be decreased by adding a codeword of the tensor code in a single local view.
The decoder does not know of such a representation, but analysing it
will be key to understanding what the decoder is able to do from the sole
observation of $\hat{Z}$, which it does see. We will denote by 
$\mathcal{Z}' = (C_0', R_0', C_1', R_1')$ the representation of the updated
value of $\hat{Z}'$ of $\hat{Z}$ after the first parallel decoding step.
This representation is deduced from $\mathcal{Z}$ through local updates as
specified by the above description of the algorithm during the first parallel
decoding step.
We note that $\mathcal{Z}'$ does not have to be a minimal representation. 
Similarly, we define $\mathcal{Z}'' = (C_0'', R_0'', C_1'', R_1'')$ to be the
representation of the updated mismatch $\hat{Z}''$ after the second parallel
decoding step, again deduced from $\mathcal{Z}''$ through the specified updates
of $\hat{R}_i$, $\hat{C}_j$.

Let us denote by $S_0=S_{00}\cup S_{11}$ the set of vertices $v\in V_{ii}$ of the
stalled sequential decoder such that
$R_i+C_i$ is non-zero on $Q(v)$. Note that $S_0$ depends on the choice
$\mathcal{Z}$ of the representation of $\hat{Z}$. We call these vertices the
{\em active vertices} of $\mathcal{Z}$ (or simply of $\hat{Z}$ when its
representation is implicit).

It will be useful to keep in mind that for any
vertex $v$ of $V_{ii}\subset V_0$, the fact that the sequential decoder cannot
decrease $|\hat{Z}|$ locally means that on $Q(v)$, the Hamming weight of
$(R_0+C_0)\cap (R_1+C_1)$ is at least half the Hamming
weight of $R_i+C_i$:
otherwise we could decrease $|\hat{Z}|$ by updating $\hat{Z}$ to $\hat{Z}+r_v+c_v$ where
$r_v+c_v$ is the value of $R_i+C_i$ on $Q(v)$.

We first start by drawing consequences from this fact that will tell us a lot
about the structure of $S_0$ and the local views of $R_0+C_0$ and
$R_1+C_1$ on the $Q$-neighbourhoods of $V_0$.

Let $0<\beta<1/2$.
Let us define $\col_0^\beta$ to be the set of columns in all $Q$-neighbourhoods of
$V_{00}$ (we could also view them in $Q$-neighbourhoods of $V_{10}$),
that
share at least $\beta\Delta$ coordinates with $C_1$, \textit{i.e.}\ that have at least
$\beta\Delta$ coordinates on which $C_1$ is equal to $1$. 
Similarly, we define $\col_1^\beta$ to be the set of columns in all $Q$-neighbourhoods of
$V_{11}$ that share at least $\beta\Delta$ coordinates with $C_0$.
We also define $\row_0^\beta$ to be the set of rows in all $Q$-neighbourhoods of
$V_{00}$ that share at least $\beta\Delta$ coordinates with $R_1$ and
likewise $\row_1^\beta$ denotes the set rows that share at least $\beta\Delta$
coordinates with $R_0$.

Our main technical
lemma reads:

\begin{lemma}
\label{lem:main}
Suppose
$|S_0|\leq \frac{\delta}{8\Delta}|V_0|$.
For $\Delta$ large enough,
There is a constant $\kappa(\beta)$, depending only on $\beta$, such that, whenever
the decoder is stalled: for $i=0,1$,
\begin{align*}
|\col_i^\beta|\leq\frac{\kappa(\beta)}{\Delta^{1-2\eps}}|S_0|.\\
|\row_i^\beta|\leq\frac{\kappa(\beta)}{\Delta^{1-2\eps}}|S_0|.
\end{align*}
\end{lemma}

A word of comment may be helpful at this point. Consider a column (say) in some
local view, of a vertex of $V_{00}$ (without loss of generality). Suppose this
column supports a non-zero $C_A$-codeword that contributes to the makeup of the global $C_0$ vector
in the current $\mathcal{Z}$-decomposition of $\hat{Z}$. What prevents us from
adding this codeword to $\hat{Z}$ and decrease its weight? It is naturally the
presence of contributions of either $R_0$, or $R_1$, or $C_1$ (or all three) to the support of
this column vector. What we want to point out here, is that the contribution of
$R_0$ is completely described by a unique $Q$-neighbourhood of a single vertex
of $V_{00}$, and that the contribution of $R_1$ is similarly described by the
$Q$-neighbourhood of just one vertex of $V_{10}$. However, the contribution of
$C_1$ stems from up to $\Delta$ different $Q$-neighbourhoods, and its nature is
very much different. What Lemma~\ref{lem:main} essentially says is that this contribution is
going to be negligible and that it will suffice for us (and the decoder) to focus on the
contributions of $R_0$ and $R_1$.

The proof of Lemma~\ref{lem:main} will rely upon the subgraphs, both of
$\G_0^\square$ and of $\G_A$, induced by the active vertices, and rely upon
expansion in both these graphs. 
We first need to define $S_{ii}^\ell$ to be the subset of $S_{ii}$ that consists
of either
\begin{itemize}
\item[(a)] the set of vertices $v$ for which the local view on $Q(v)$ of
$R_i+C_i$
is made up of at least $\ell$ non-zero columns vectors of $C_i$,
\item[(b)] or the set of vertices $v$ for which the local view on $Q(v)$ of
$R_i+C_i$
is made up of at least $\ell$ non-zero row vectors of $R_i$.
\end{itemize}

\begin{lemma}
\label{lem:Sl}
For $i=0,1$, if $|S_0|\leq \frac{\delta}{8\Delta}|V_0|$,
we have 
\begin{eqnarray*}
|S_{ii}^\ell| \leq &\frac{256}{\delta^2\ell^2}|S_0|\quad&\text{if
$\delta\ell\leq\Delta^{1/2-\eps}$}\\
|S_{ii}^\ell| \leq &\frac{256}{\Delta^{1-2\eps}}|S_0|\quad&\text{if
$\delta\ell\geq\Delta^{1/2-\eps}$}.
\end{eqnarray*}
\end{lemma}

\begin{proof}
We prove the result for $i=1$ and definition (a) of $S_{11}^\ell$, the other
cases being identical.
We consider
$E(S_{11}^\ell,S_{00})$ in the subgraph of
$\G^\square$ where the edges
belong to $(R_0+C_0)\cap(R_1+C_1)$. We argue that the degree of every vertex of
$S_{11}$ is at least half the weight of its local view in $R_1+C_1$, so at least
half of
$\ell\delta\Delta$ when this quantity is not more than
$\Delta^{3/2-\eps}$, by robustness of the local dual tensor code. This gives 
\[
\frac 12\ell\delta\Delta |S_{11}^\ell|\leq |E(S_{11}^\ell,S_{00})|.
\]
We now apply the Expander mixing Lemma~\ref{lem:mixing} in $\G_0^\square$.
This gives
\begin{align*}
\frac 12\ell\delta\Delta|S_{11}^\ell| &\leq
\Delta^2\frac{|S_{11}^\ell||S_{00}|}{|V_{00}|}+ 4\Delta\sqrt{|S_{11}^\ell||S_{00}|}\\
&\leq \frac 14\delta\Delta|S_{11}^\ell| + 4\Delta\sqrt{|S_{11}^\ell||S_{00}|}\\
&\text{since $|S_{00}|/|V_{00}|\leq 2|S_0|/|V_0|$ and $|S_0| \leq \delta |V_0|/8\Delta$}\\
\ell\delta \sqrt{|S_{11}^\ell|}&\leq 16\sqrt{|S_{00}|}\\
|S_{11}^\ell|&\leq \frac{256}{\delta^2\ell^2}|S_{00}|
\end{align*}
which gives the first statement of the lemma by writing $|S_{00}|\leq |S_0|$.
The second statement follows analogously, by arguing that when $\delta\ell\geq\Delta^{1/2-\eps}$
the weight of the local view is not less than $\Delta^{3/2-\eps}$.
\end{proof}

\begin{lemma}
\label{lem:tbl}
Let $T_{10}^{\beta,\ell}$ be the set of vertices of $V_{10}$ whose
$Q$-neighbourhoods contain at least
$\ell$ columns of $\col_0^\beta$. Under the hypothesis
$|S_0|\leq \frac{\delta}{8\Delta}|V_0|$,
for $\Delta$ large enough,
we have
\[
|T_{10}^{\beta,\ell}|\leq \frac{64}{\beta^2\Delta}|S_{11}^{\ell\beta/2}|.
\]
\end{lemma}

\begin{proof}
Consider a the subarray of the $Q$-neighbourhood of a vertex of
$T_{10}^{\beta,\ell}$ consisting of the $\ell$ columns of $\col_0^\beta$. 
One easily checks, by counting the number coordinates in the subbarray on which
$\hat{C}_1$ must be equal to $1$, that there are at least $\beta\Delta/2$ rows on
which $\hat{C}_1$ has weight at least $\beta\ell/2$. This in turn implies that, in the
graph $\G_A$, this vertex has at least $\beta\Delta/2$ outgoing edges that fall into
$|S_{11}^{\ell\beta/2}|$.

Let us estimate $|E(T_{10}^{\beta,\ell},S_{11}^{\ell\beta/2})|$ in $\G_A$. We
have just proved
\[
|T_{10}^{\beta,\ell}|\frac 12\beta\Delta\leq |E(T_{10}^{\beta,\ell},S_{11}^{\ell\beta/2})|.
\]
We now apply the Expander mixing Lemma~\ref{lem:mixing} in $\G_A$.
We get
\[
|T_{10}^{\beta,\ell}|\frac 12\beta\Delta \leq
\Delta\frac{|T_{10}^{\beta,\ell}||S_{11}^{\ell\beta/2}|}{|V_{11}|} +
2\sqrt{\Delta}\sqrt{|T_{10}^{\beta,\ell}||S_{11}^{\ell\beta/2}|}.
\]
Writing $|S_{11}^{\ell\beta/2}|/|V_{11}|\leq |S_{11}|/|V_{11}|\leq
2|S_0|/|V_0|$, we get, for fixed $\beta$ and $\Delta$ large enough,
\begin{align*}
|T_{10}^{\beta,\ell}|\frac 14\beta\Delta &\leq
2\sqrt{\Delta}\sqrt{|T_{10}^{\beta,\ell}||S_{11}^{\ell\beta/2}|}\\
|T_{10}^{\beta,\ell}|^{1/2}&\leq \frac{8}{\beta\Delta^{1/2}}|S_{11}^{\ell\beta/2}|^{1/2}
\end{align*}
whence the claimed result after squaring.
\end{proof}

\begin{proof}[Proof of Lemma~\ref{lem:main}]
We prove the Lemma for $|\col_0^\beta|$, the other cases being essentially
identical. Write
\begin{align}
\label{eq:penible}
|\col_0^\beta| \leq \sum_{\ell<\ell_{\max}}\ell | T_{10}^{\beta,\ell}|
+\Delta|T_{10}^{\beta,\ell_{\max}}|
\end{align}
where $\ell_{\max}$ is chosen so as to have
$\delta\ell_{\max}\beta/2=\Delta^{1/2-\eps}$. Applying Lemmas~\ref{lem:tbl}
and~\ref{lem:Sl} we obtain
\begin{align*}
|\col_0^\beta| &\leq \sum_{\ell<\ell_{\max}}\ell \frac{64}{\beta^2\Delta}|S_{11}^{\ell\beta/2}|
+\Delta\frac{64}{\beta^2\Delta}|S_{11}^{\ell_{\max}\beta/2}|\\
&\leq \sum_{\ell}\frac{2^{16}}{\delta^2\ell\beta^4}\frac{1}{\Delta}|S_0| +
\frac{2^{14}}{\delta^2\beta^2\Delta^{1-2\eps}}|S_0|
\end{align*}
from which the result follows, after writing that
$\sum_{\ell}1/\ell\leq\ln\Delta\leq\Delta^{2\eps}$.
\end{proof}

\begin{definition}[normal vertices of $S_{00}$ and $S_{11}$] Let us say that a vertex $v$ of $S_{00}$ is
\emph{normal} if $\|r_v\|+\| c_v\| < \Delta^{1/2-\eps}/4$, where $r_v,c_v$ are the local
components of $R_0,C_0$. Normal vertices of $S_{11}$ are defined similarly. A
vertex which is not normal will be said to be \emph{exceptional}.
\end{definition}
Below we will talk about normal columns of $C_0$ and normal rows of
$R_1$. This
will mean that in their respective local views in $V_{00}$ and $V_{11}$ they
belong to the local view of a normal vertex.

The terminology ``normal'' and ``exceptional'' is justified by the following
lemma.

\begin{lemma}
\label{lem:Se}
Let $S_0^e$ be the set of exceptional vertices of $S_0$. Under the hypothesis
$|S_0|\leq\frac{\delta}{8\Delta}|V_0|$,
we have 
\[
|S_0^e|\leq \frac{\kappa}{\delta^2\Delta^{1-2\eps}}|S_0|
\]
for some absolute constant $\kappa$.
\end{lemma}

\begin{proof}
Let $S_{00}^e$ ($S_{11}^e$) be the set of exceptional vertices of $S_{00}$
($S_{11}$). The definition of an exceptional (not normal)
vertex $v$ of $S_{00}$ implies that $R_0+C_0$ has weight at least
$\delta\Delta^{3/2-\eps}/8$ on $Q(v)$, and therefore that $(R_0+C_0)\cap
(R_1+C_1)$ has weight at least $\delta\Delta^{3/2-\eps}/16$ on $Q(v)$. The weight
of $(R_0+C_0)\cap (R_1+C_1)$ on $\bigcup_{v\in S_{00}^e}Q(v)$ is therefore at
least 
\[
|S_{00}^e|\delta\Delta^{3/2-\eps}/16
\]
and it is also not more than the number of edges $|E(S_{00}^e,S_{11})|$ in the
graph $\G_0^\square$. By the Expander mixing Lemma in $\G_0^\square$,
we have
\[
\frac{1}{16}\delta\Delta^{3/2-\eps}|S_{00}^e|\leq |E(S_{00}^e,S_{11})| \leq \Delta^2\frac{|S_{00}^e||S_{11}|}{|V_{00}|} +
4\Delta\sqrt{|S_{00}^e||S_{11}|}.
\]
Writing $|S_{11}|/|V_{00}|\leq 2|S_0|/|V_0|$, for $\Delta$ large enough, we straightforwardly obtain that
\begin{align*}
	\frac{1}{32} \delta\Delta^{3/2-\eps}|S_{00}^e|^{1/2}&\leq
4\Delta|S_{11}|^{1/2}\\
|S_{00}^e| &\leq \frac{2^{14}}{\delta^2}\frac{1}{\Delta^{1-2\eps}}|S_{11}|.
\end{align*}
We obtain similarly
\[
|S_{11}^e| \leq \frac{2^{14}}{\delta^2}\frac{1}{\Delta^{1-2\eps}}|S_{00}|.
\]
Hence the result, by summing the two inequalities.
\end{proof}

Similarly to $S_{0}=S_{00}\cup S_{11}$, let us set $S_1=S_{10}\cup S_{01}$ where
$S_{10}$ and $S_{01}$ are the set of active vertices of $V_{10}$ and $V_{01}$
relative to $\hat{Z}$. In other words, $S_{10}$ $(S_{01})$ is the set of
vertices of $V_{10}$ ($V_{01}$)
on which the local view of $R_1+C_0$ ($R_0+C_1$) is
non-zero.

\begin{definition}[agglutinating vertices of $V_{ij}$, $i\neq j$] Let us say
that a vertex $v$ of $S_{ij}$ is
{\em agglutinating} if $\|r_v\|\geq \delta\Delta/3$, where $r_v$ is the
local component of $R_i$, or if $\|c_v\|\geq \delta\Delta/3$, where $c_v$ is the
local component of $C_j$.
\end{definition} 
\begin{definition}
[decodable agglutinating vertex] Let us say that an agglutinating vertex $v$ of
$S_{10}$ (resp. $S_{01}$) is {\em decodable}, if at most $\Delta^\gamma/2$ columms in its local
view are in common with exceptional vertices of $S_{00}$ (resp. $S_{11}$), and at most
$\Delta^\gamma/2$ rows in its local view are in common with exceptional vertices
of $S_{11}$ (resp. $S_{00}$).
\end{definition}

We now show that every decodable agglutinating vertex really is decodable in the
sense of the first parallel decoding step:

\begin{proposition}\label{prop:agglu}
The tweaked decoder that is applied during the first parallel decoding step,
described in Section~\ref{sec:tweaked}, identifies correctly every decodable
agglutinating vertex $v$ of $S_{10}$ or of $S_{01}$, \textit{i.e.}\ it finds
the required non-zero $r_v+c_v$. 
\end{proposition}

\begin{proof}
Let $v$ be a decodable agglutinating vertex of $S_{10}$, the case $v\in S_{01}$
being essentially identical.
We set $A_0$ to consist of all the indexes of all normal rows in $Q(v)$ and $B_0$ to consist of all
the indexes of normal columns. We now show that
the local component $r_v+c_v$ of $R_1+C_0$, restricted to
$A_0\times B_0$
satisfies all the requirements of the tweaked decoder. Indeed, we must have
$|(r_v+c_v)_{|A_0\times B_0}|>w$, otherwise
robustness (with resistance to puncturing)
would imply that $r_v+c_v$ would be expressible as sum of fewer than
$O(\Delta^{1/2})$ non-zero rows and columns of $C_A\otimes\f_2^B$ and
$\f_2^A\otimes C_B$, which would contradict local minimality of
$\|r_v\|+\|c_v\|$, since we have supposed either $\|r_v\|\geq \delta\Delta/3$,
or $\|c_v\|\geq \delta\Delta/3$, by
definition of an agglutinating vertex.

Furthermore, the contribution of $R_0$ on any normal column is less than $\Delta^{1/2-\eps}/4$ by definition of a
normal column. Therefore, the weight of $R_0$ on $A\times B_0$ is less than
$\Delta\cdot\Delta^{1/2-\eps}/4=w/4$. Similarly, the weight of $C_1$ on $A_0\times B$
is less than $w/4$, and therefore the weight of $R_0+C_1$ on
$A_0\times B_0$ is
$<w/2$. After removing $r_v+c_v$, the updated value $R_1'+C_0'$ on
$A_0\times B_0$ is
zero, so the value of $|\hat{Z}|$ on $A_0\times B_0$ is that of
$|R_0+C_1|$ and is
$<w/2$.

Finally, since the weight of $R_0+C_1$ on $A_0\times B_0$ is $<w/2$ and the weight
of $R_1+C_0$ on $A_0\times B_0$ is $>w$, we have that the weight of
$\hat{Z}=R_0+C_1+R_1+C_0$
on $A_0\times B_0$ is $>w/2$.
\end{proof}

We remark that the above proof tells us in passing that the Hamming weight of $\hat{Z}$
on the local view of a decodable agglutinating vertex is at least $w/2$. We note
this for future reference.

\begin{proposition}\label{prop:Zagglu}
Whenever the sequential decoder is stalled, the Hamming weight of $\hat{Z}$
restricted to $Q(v)$, for $v$ a decodable agglutinating vertex of $S_{10}$ or of
$S_{01}$, is at least equal to $\Delta^{3/2-\eps}/2$.
\end{proposition}

Let us denote by $K^\gamma_{10}$ ($K_{01}^\gamma$) the set of vertices of
$V_{10}$ ($V_{01}$) whose
$Q$-neighbourhoods have at least
$\Delta^\gamma/2$ columns (rows) in common with exceptional vertices of $S_{00}$, or that
have at least $\Delta^\gamma/2$ rows (columns) in common with exceptional vertices of
$S_{11}$. We then set $K^\gamma=K_{10}^\gamma \cup K_{01}^\gamma$. These can be
thought of as the vertices of $V_1$ whose $Q$-neighbourhoods share an exceptionally large number of
rows or of columns with exceptional vertices of $S_0$.

\begin{lemma}
\label{lem:Tgamma}
Suppose $|S_0|\leq \frac{\delta}{8\Delta}|V_0|$.
Then there is an absolute constant $\kappa$ such that 
\[
|K^\gamma|\leq
\frac{\kappa}{\delta^2}\frac{1}{\Delta^{2-4\eps}}|S_0|.
\]
\end{lemma}
\begin{proof}
We give an upper bound for $K_{10}^\gamma$, an identical one follows for
$K_{01}^\gamma$ by symmetry. Consider the case when $K_{10}^\gamma$ has $\Delta^\gamma/2$ columns 
in common with the $Q$-neighbourhoods of exceptional vertices of $S_{00}^e$,
the other case being similar. This last fact means that 
all vertices of $K_{10}^\gamma$ have degree at least
$\Delta^\gamma/2$ in the subgraph of $\G_B$ induced by $K_{10}^\gamma$ and
$S_{00}^e$.
By the Expander mixing Lemma in $\G_B$ we therefore have
\[
|K_{10}^\gamma|\Delta^\gamma/2\leq |E(K^\gamma,S_0^e)|\leq
\Delta\frac{|K_{10}^\gamma||S_{00}^e|}{|V_{00}|} + 2\sqrt{\Delta}\sqrt{|K_{10}^\gamma||S_0^e|}.
\]
Bounding $|S_{00}^e|/|V_{00}|$ from above by $2|S_0|/|V_0|$ we easily get
\[
|K_{10}^\gamma|\Delta^\gamma/4\leq 2\sqrt{\Delta}\sqrt{|K_{10}^\gamma||S_0^e|}
\]
whence
\begin{align*}
|K_{10}^\gamma|^{1/2} &\leq 8\Delta^{1/2-\gamma}|S_0^e|^{1/2}\\
|K_{10}^\gamma|       &\leq 64\Delta^{1-2\gamma}|S_0^e|\\
                 &\leq\frac{\kappa}{\delta^2}\Delta^{2\eps-2\gamma}|S_0|
\end{align*}
by applying Lemma~\ref{lem:Se}. Hence the result, since we have supposed $\gamma
= 1-\eps$.
\end{proof}

\begin{proposition}\label{prop:decoded}
Let $v$ be a vertex of $S_{10}$ (resp. $S_{01}$) that is not in $K^\gamma$ and that is decoded by
the tweaked decoder. Let $A_n$ and $B_n$ be the index
sets of its normal rows and its normal columns respectively. After (tweaked)
decoding, the contribution of $R_1'+C_0'$ (resp.
$R_0'+C_1'$) to the coordinate set indexed by $A_n\times B_n$ in
$Q(v)$ is supported by
columns of $\col_0^\beta$ and $\row_1^\beta$ (resp. $\col_1^\beta$ and
$\row_0^\beta$) for some $\beta >1/3$.
\end{proposition}

\begin{proof}
We treat the case $v\in S_{10}$, the case $v\in S_{01}$ being essentially
identical.
First consider the dual-tensor codeword $r_v'+c_v'$ 
that the decoder has
identified on $A_0\times B_0$ (which may not be equal to $A_n\times B_n$).
We compare it to $r_v+c_v$, the reduced local view of $R_1+C_0$. 
We recall, that by definition of the tweaked decoder, we must have 
that  $|\hat{Z}+r'_v+c'_v|$ is less than $w/2$ on $A_0\times B_0$.
So in particular it must be less than $w/2$ on $(A_0\cap A_n)\times (B_0\cap
B_n)$. Since the value of $|\hat{Z}+r_v+c_v|$ is less than $w/2$ on $A_n\times B_n$,
by definition of normal rows and columns, we must have that
$|r_v+c_v+r'_v+c'_v|\leq w$ on $(A_0\cap A_n)\times (B_0\cap B_n)$. 
Now $A_0\cap A_n$ and $B_0\cap B_n$ both have cardinality at least 
$\Delta -\Delta^\gamma/2-\Delta^\gamma/2=\Delta-\Delta^\gamma$, so that we may
apply robustness with resistance to puncturing to $r_v+c_v+r'_v+c'_v$
and deduce that $\|r_v\|+\|c_v\|+\|r'_v\|+\|c'_v\|$ is $O(\Delta^{1/2} + \Delta^\gamma)$. 
This in turn implies that, after tweaked decoding,
every normal column (\textit{i.e.}\ indexed by $b\in B_n$) has at most $O(\Delta^{1/2}+\Delta^\gamma)$ coordinates on
which $R_0$ or $R_1'$ equals $1$: therefore, if a column, indexed by
$b\in B_n$ say, is the
support of a non-zero column codeword $c\in C_A$ in $C_0'$, the only way this
column has not been decoded is for the support of $c$ on $A\times\{b\}$  to coincide with the support of
$C_1$ on almost half its coordinates (minus
$O(\Delta^{1/2}+\Delta^\gamma)$), otherwise more than half the support of $c$
contributes to $\hat{Z}$, and adding $c$ would decrease $|\hat{Z}|$.
Summarising, 
if the column under discussion does not belong to
$\col_0^\beta$, it will be decoded by the tweaked decoder. A similar argument
holds for normal rows.
\end{proof}

Note that the choice $\beta >1/3$ in this last proposition is somewhat arbitrary
and could in principle be chosen arbitrarily close to $1/2$.

\begin{proposition}\label{prop:notdecoded}
Let $v$ be a vertex of $S_{10}$ (resp. $S_{01}$) that is not in $K^\gamma$ and that is not identified
as decodable (hence not decoded) by the tweaked decoder. Let $A_n$ and $B_n$ be the index
sets of its normal rows and its normal columns respectively. Then the
contribution of $R_1+C_0$ (resp. $R_0+C_1$) to the
coordinate set indexed by $A_n\times B_n$ in $Q(v)$ is supported by columns of
$\col_0^\beta$ and $\row_1^\beta$ (resp. $\col_1^\beta$ and $\row_0^\beta$) for
$\beta =1/6+o_\Delta(1)$.
\end{proposition}

\begin{proof}
Suppose $v\in S_{10}$ is such a vertex.
The vertex $v$ cannot be agglutinating by Proposition~\ref{prop:agglu}.
Indeed, if it were  agglutinating, then since we suppose it not in $K^\gamma$, it
would be a decodable agglutinating vertex and would be decoded by the tweaked
decoder.
Therefore, we have $\|r_v\|<\delta\Delta/3$ and $\|c_v\|<\delta\Delta/3$, where
$r_v,c_v$ are the components of $R_1$ and $C_0$ on $Q(v)$.
For every column of $Q(v)$, the contribution of $R_1$ to its support is therefore $<\delta\Delta/3$.
For every normal column, the contribution of $R_0$ to its support is
$<\Delta^{1/2-\eps}/4$ by definition of normality. Therefore, on every
normal column on which $C_0$ is non-zero we must have a contribution of
$C_1$
that is at least $\delta\Delta/2-\delta\Delta/3-\Delta^{1/2-\eps}/4$.
We argue similarly for normal rows in $Q(v)$.
\end{proof}

\paragraph{Situation after the first parallel decoding step.}
We now argue that this tweaked decoding step has drastically reduced the number
of active vertices of $\hat{Z}$ in $V_0$. 

A remaining active vertex of $V_0$ may be one of the exceptional vertices of
$S_0^e\subset S_0$. By
Lemma~\ref{lem:Se} there are at most $\kappa|S_0|/\Delta^{1-2\eps}$ of them, for some
constant $\kappa$.

To bound from above the additional number of active vertices, we simply bound
from above the number of $C_A$-column vectors of $C_0'$ and $C_1'$ 
(and of $C_B$-row vectors of $R_0'$ and $R_1'$)
that are supported by the $Q$-neighbourhood of a non-exceptional vertex
$v\notin S_e$.
Such a row or column codeword may
\begin{itemize}
\item[--] belong to the $Q$-neighbourhood of a vertex of $V_1$ that has many
columns or rows in common with exceptional vertices of $V_0$: the set of such
vertices was
defined as $K^\gamma$, and by Lemma~\ref{lem:Tgamma}, there are not more than 
$\kappa|S_0|/\Delta^{2-4\eps}$ of them. These vertices therefore contribute at most 
$2\Delta|K^\gamma|$ row and column vectors, so not more than
$\kappa|S_0|/\Delta^{1-4\eps}$ for some constant $\kappa$. 
\item[--] Remaining row or column codewords (that have not been removed by the
decoder or have wrongly been introduced by the decoder) may also be present in
the $Q$-neighbourhoods of vertices of $V_1$ that have been decoded by the
tweaked decoder. By Proposition~\ref{prop:decoded}, all these rows and columns
must belong to $\row_i^\beta$ and $\col_i^\beta$, for some fixed value $\beta$
and $i=0,1$. By Lemma~\ref{lem:main} there are at most $\kappa|S_0|/\Delta^{1-2\eps}$
such column and row vectors, for some constant $\kappa$.
\item[--] The remaining row or column codewords must be present in the the
$Q$-neighbour\-hoods of vertices of $V_1$ 
that are not in $K^\gamma$ and have not been decoded by the tweaked decoder. By
Proposition~\ref{prop:notdecoded}, all these rows and columns must belong to $\row_i^\beta$ and $\col_i^\beta$, for some fixed value $\beta$
and $i=0,1$. Again, by Lemma~\ref{lem:main} there are at most $\kappa|S_0|/\Delta^{1-2\eps}$
such column and row vectors, for some constant $\kappa$.
\end{itemize}

Summing everything, we have shown the following:

\begin{theo}
\label{thm:tweaked}
We suppose
$|S_0|\leq \frac{\delta}{8\Delta}|V_0|$.
Let $S_0'$ be the set of active vertices of $V_0$ after the first parallel
decoding step, \textit{i.e.}\ the set of vertices of $V_{00}$ on whose $Q$-neighbourhood
$R_0'+C_0'$ is non-zero, together with the set of vertices of $V_{11}$
on on whose $Q$-neighbourhood $R_1'+C_1'$ is non-zero.
We have $|S_0'|\leq \kappa|S_0|/\Delta^{1-4\eps}$ for some constant $\kappa$.
\end{theo}

Before analysing the second parallel decoding step, we need the following
consequence of the above analysis.

\begin{proposition}\label{prop:stalled}
When the sequential decoder 
is stalled, we have $|\hat{Z}|\geq \kappa|S_0|\Delta^{1/2-\eps}$ for some constant $\kappa$.
\end{proposition}

\begin{proof}
By Lemma~\ref{lem:Se}, the number of exceptional vertices in $S_0^e$ is a small fraction of
the set of active vertices $S_0$. There are therefore at least at least $|S_0|/2$
(say) normal rows and columns. By Lemma~\ref{lem:Tgamma} only a small fraction
of the set of normal rows and columns spanned by $S_0$ can end up in the $Q$-neighbourhood of a
vertex of $K^\gamma$ in $V_1$. Furthermore, 
by arguing as in the proof of
Proposition~\ref{prop:notdecoded}, if a normal row or a normal column is not in
the $Q$-neighbourhood of a decodable agglutinating vertex, it must belong to one of the
sets $\row_i^\beta$ or $\col_j^\beta$ for $\beta=1/6+o_\Delta(1)$, otherwise it
would have already been decoded by the sequential decoder.
Therefore, most normal rows and columns of $S_0$ must belong to the
$Q$-neighbourhood of a decodable agglutinating vertex of $S_1$ (hence the terminology
'agglutinating'). The number of decodable agglutinating vertices of $S_1$ is therefore at
least equal to $\kappa|S_0|/\Delta$ for some constant $\kappa$, and the result follows from
Proposition~\ref{prop:Zagglu}.
\end{proof}

\paragraph{Analysis of the second parallel decoding step.}
This step is simpler than the first. Recall that the second parallel decoding
step is applied on vertices of $V_0$ and uses regular decoding (\textit{i.e.}\ decreases
$|\hat{Z}|$ as much
as possible on complete, unpunctured, local views). 

To avoid confusion between the different stages of the decoder, let us denote by
$\hat{Z},\hat{Z}',\hat{Z}''$ respectively the values of the mismatch $\hat{Z}$ of the stalled
serial decoder before the
first parallel decoding step,
just after the first parallel decoding step on $V_1$, and finally after the
second parallel decoding step on $V_0$. Accordingly we write
$\hat{Z}=C_0+R_0+C_1+R_1$, $\hat{Z}'=C'_0+R'_0+C'_1+R'_1$, 
$\hat{Z}''=C''_0+R''_0+C''_1+R''_1$, and use the
notation $S_0=S_{00}\cup S_{11}$, $S_0'=S_{00}'\cup S_{11}'$,
$S_0''=S_{00}''\cup S_{11}''$ for the set of active vertices
of $V_0$ at each stage of decoding.

We recall also that $\mathcal{Z}=(C_0,R_0,C_1,R_1)$ has been chosen to be a minimal
representation of $\hat{Z}$. The
representations $\mathcal{Z'}=(C_0',R_0',C_1',R_1')$ and
$\mathcal{Z}''=(C_0'',R_0'',C_1'',R_1'')$ 
are however simply deduced from $\mathcal{Z}$ by the decoder's increments.

\begin{lemma}
\label{lem:S''00}
Let $x_0\geq 1$ (resp. $x_1\geq 1$) be such that $x_0\delta\Delta$ (resp.
$x_1\delta\Delta$) is the average
weight of $R''_0+C''_0$ (resp. $R''_1+C''_1$) on the union of the
$Q$-neighbourhoods of $S''_{00}$ (resp. $S''_{11}$). Under the hypothesis
$|S_0|\leq \frac{\delta}{8\Delta}|V_0|$,
we have
\begin{align*}
|S''_{00}|&\leq\frac{256}{x_0^2\delta^2}|S'_0|\\
|S''_{11}|&\leq\frac{256}{x_1^2\delta^2}|S'_0|.
\end{align*}
\end{lemma}

\begin{proof}
Observe that $x_i\geq 1$, $i=0,1$, because $\delta\Delta$ is the minimum distance of the dual tensor code. 
Now for any vertex $v$  of $V_{00}$, we must have
\begin{equation}\label{eq:basic}
|(R''_0+C''_0)_{|Q(v)}|\leq 2|(R'_1+C'_1)_{|Q(v)}|
\end{equation} 
(otherwise the zero vector
is a better choice than $R''_0+C''_0$ on $Q(v)$, since it gives a lower value of
$|\hat{Z}''|$).

From \eqref{eq:basic} we have
that the average degree of vertices of $S''_{00}$ in the subgraph of $\G_0^\square$
induced by $S''_{00}$ and $S'_{11}$ is at least $x\delta\Delta/2$. Applying the
Expander mixing Lemma in $\G_0^\square$ we have 
\[
\frac 12|S''_{00}|x\delta\Delta \leq |E(S''_{00},S'_{11})|\leq
\Delta^2\frac{|S_{00}''||S_{11}'|}{|V_{00}|} +
4\Delta\sqrt{|S''_{00}||S'_{11}|}.
\]
Writing $|S_{11}'|/|V_{00}|\leq 2|S_0'|/|V_0|\leq 2|S_0|/|V_0|$ by
Theorem~\ref{thm:tweaked} and applying the hypothesis
$|S_0|\leq\frac{\delta}{8\Delta}|V_0|$, we get
\begin{align*}
\frac 14|S''_{00}|x_0\delta\Delta&\leq 4\Delta\sqrt{|S''_{00}||S'_0|}\\
|S''_{00}|^{1/2}&\leq \frac{16}{x_0\delta}|S'_0|^{1/2}\\
|S''_{00}|&\leq\frac{256}{x_0^2\delta^2}|S'_0|.
\end{align*}
The same reasoning is valid for $S''_{11}$.
\end{proof}

The following theorem is the core result of the analysis of the whole parallel
decoding procedure, it states that together the two steps of the parallel
decoding procedure decrease the Hamming weight of $\hat{Z}$.

\begin{theo}\label{thm:hatZ}
Under the hypothesis
$|S_0|\leq \frac{\delta}{8\Delta}|V_0|$,
after the second parallel decoding step we have $|\hat{Z}''|\leq \kappa
|S_0|\Delta^{4\eps}$ for some
constant $\kappa$.
\end{theo}

\begin{proof}
From Lemma~\ref{lem:S''00} we obtain that
\[
|R''_0+C''_0|=\sum_{v\in S''_{00}}
|(R''_0+C''_0)_{Q(v)}|=|S''_{00}|x_0\delta\Delta\leq\frac{256}{x_0\delta}|S'_0|\Delta
\]
and applying Theorem~\ref{thm:tweaked} we obtain
\[
|R''_0+C''_0|\leq \kappa|S_0|\Delta^{4\eps}
\]
for some constant $\kappa$ (remembering that $x_0\geq 1$).
We have naturally the same upper bound for $|R''_1+C''_1|$
and since
\[
|\hat{Z}''|\leq |R''_0+C''_0| + |R''_1+C''_1|
\]
we have the claimed result.
\end{proof}

For the proof of Theorem~\ref{thm:main} we will also need the following Lemma. Recall that
$\|\hat{Z}\|$ denotes the minimum weight $\|\mathcal{Z}\|$ of a representation
of $\hat{Z}$.

\begin{lemma}
\label{lem:/2}
Let $\hat{Z}$ and $\hat{Z}''$ be the states of $\hat{Z}$ just before a parallel
decoding procedure and just after the second step of the parallel decoding
procedure. We have $\|\hat{Z}''\|\leq \|\hat{Z}\|/2$.
\end{lemma}

\begin{proof}
Recall that we have taken
$\mathcal{Z}=(C_0,R_0,C_1,R_1)$ such that $\|\mathcal{Z}\|=\|\hat{Z}\|$.
As mentioned just before Lemma~\ref{lem:S''00}, the two
representations $\mathcal{Z'}=(C_0',R_0',C_1',R_1')$ and
$\mathcal{Z}''=(C_0'',R_0'',C_1'',R_1'')$ that are deduced from $\mathcal{Z}$
through the parallel decoder's updates are not necessarily minimal, but
nevertheless it suffices to show that $\|\mathcal{Z}''\|\leq \|\mathcal{Z}\|/2$.

Let $x_0,x_1$ be defined as in Lemma~\ref{lem:S''00}.
We consider two cases
\begin{itemize}
\item[(i)] $x_0>\Delta^{1/4}$. In this case Lemma~\ref{lem:S''00}
and Theorem~\ref{thm:tweaked} imply
$|S_{00}''|\leq\frac{\kappa\Delta^{4\eps}}{\delta^2\Delta^{3/2}}|S_0|$ for some constant
$\kappa$. We deduce $\|R_0''\|+\|C_0''\|\leq |S_{00}''|2\Delta\leq |S_0|/4$ for $\Delta$ large
enough.
\item[(ii)] $x_0\leq\Delta^{1/4}$. All vertices of $S_{00}''$ on which
$R_{00}''+C_{00}''$ has Hamming weight less than $\Delta^{3/2-\eps}$ contribute, by
applying robustness with $w=\Delta^{3/2-\eps}$, at most
$\frac{2}{\delta}\Delta^{1/2-\eps}|S_{00}''|$ to $\|R_0''\|+\|C_0''\|$. The
vertices of $S_{00}''$ on which $R_{00}''+C_{00}''$ has Hamming weight $\geq \Delta^{3/2-\eps}$
number not more than $\frac{\delta}{\Delta^{1/4-\eps}}|S_{00}''|$, and we again easily
deduce from Lemma~\ref{lem:S''00}
and Theorem~\ref{thm:tweaked} that $\|R_0''\|+\|C_0''\|\leq |S_0|/4$ for $\Delta$ large
enough.
\end{itemize}
We therefore always have $\|R_0''\|+\|C_0''\|\leq |S_0|/4$. We also have the
symmetrical inequality
$\|R_1''\|+\|C_1''\|\leq |S_0|/4$, and summing both we get 
\[
\|R_0''\|+\|C_0''\|+\|R_1''\|+\|C_1''\|\leq |S_0|/2\leq \|\hat{Z}\|/2.
\]
\end{proof}

\begin{proof}[Proof of Theorem~\ref{thm:main}]

Let $\error$ be the original error vector, with local view $e_v$ for vertices $v\in
V$. We recall that the decoder initially
computes a local estimation $\eps_v$ of the error on all $Q$-neighbourhoods of
vertices $v\in V_1$. We have $\eps_v=e_v+r_v+c_v$ where $r_v+c_v$ is a local
dual tensor codeword.
Over the partition $(Q(v))_{v\in V_{01}}$ of the set $Q$ of coordinates, 
the global vector $\beps_{01}$ defined by the local views $\eps_v$ differs from the error
vector $\error$ by a vector $R_0+C_1$, similarly the global vector $\beps_{10}$ defined by the local views $\eps_v$ 
for $v\in V_{10}$ differs from $\error$ by a vector $R_1+C_0$. The two vectors
$R_0+C_1$ and $R_1+C_0$ do not a priori coincide, and their sum defines the
original mismatch $Z$. The goal of the decoder is to find a decomposition
$Z=\hat{C}_0 + \hat{R}_0+\hat{C}_1+\hat{R}_1$ and output
$\hat{\error}=\beps_{01}+\hat{R}_0+\hat{C}_1$.

We will have proved that our decoder works after ensuring two things:

\begin{itemize}
\item[--] that it indeed finds a decomposition
$Z=\hat{C}_0 + \hat{R}_0+\hat{C}_1+\hat{R}_1$,
\item[--] that the corresponding $\hat{\error}$ is such that
$|\error +\hat{\error}|=|R_0+C_1+\hat{R}_0+\hat{C}_1|<d_{\min}$ where
$d_{\min}$ is the minimum distance of the quantum code.
\end{itemize}

By Proposition~\ref{prop:stalled} and Theorem~\ref{thm:hatZ} we have that, for a
sufficiently small value of $\eps$ (namely $\eps <1/10$), the
Hamming weight $|\hat{Z}|$ decreases during parallel decoding steps whenever we have
$|S_0|\leq \frac{\delta}{8\Delta}|V_0|$. Furthermore, the
Hamming weight $|\hat{Z}|$ decreases at every sequential
decoding step by definition. Therefore to ensure convergence to some valid
output $\hat{\error}$ we only need to ensure that the condition $|S_0|\leq \frac{\delta}{8\Delta}|V_0|$
is satisfied throughout the decoding process.

To estimate the
number $|S_0|$ of active vertices in $V_0$, we can simply upper bound it by
$\|\hat{Z}\|$ and track the evolution of $\|\hat{Z}\|$ during the course of the
decoding algorithm.

At the initial stage, we have $\hat{Z}=Z$ which is the original mismatch vector.
We note that we will have $r_v+c_v\neq 0$ only if
$|e_v|\geq\delta\Delta/2$. Since $Q$ is partitioned into $Q$-neighbourhoods of
vertices of $V_{01}$, we have that the number of vertices of $V_{01}$ for which
 $r_v+c_v\neq 0$ is at most $2|\error|/\delta\Delta$, with the same conclusion
holding for $V_{01}$. Crudely upper bounding the contribution of every such
vertex to $\|Z\|$ by $2\Delta$, we have $\|Z\|\leq 8|\error|/\delta$. Strictly
speaking, we have upper bounded the number of terms of one possible expression
of $Z$ as a sum of local row and column vectors, but the conclusion holds anyway
since $\|Z\|$ refers to the minimum number of terms of such a decomposition of
$Z$.

Similarly, we remark that $|r_v+c_v|\leq 2|e_v|$,  from which we have the
upper bound for the Hamming weight of $Z$, $|Z|\leq 4|\error|$, and also 
\begin{equation}
\label{eq:<2|e|}
|R_0+C_1|\leq 2|\error|,
\end{equation}
for $(R_0,C_0,R_1,C_1)$ the original decomposition of $Z$.

For any sequential decoding step, we clearly have that the weight $\|\hat{Z}\|$
is increased by at most $2\Delta$. Since the weight $\|\hat{Z}\|$ only decreases
during a parallel decoding step, we have that $\|\hat{Z}\|$ stays upper bounded
by $8|\error|/\delta + |Z|2\Delta \leq 8|\error|(\Delta+1/\delta)$ so that the requirement
$|S_0|\leq \frac{\delta}{8\Delta}|V_0|$
is always satisfied whenever  
\[
8|\error|(\Delta+\frac{1}{\delta})\leq \frac{\delta}{8\Delta}|V_0|
\]
\textit{i.e.}\ whenever
\[
|\error|\leq\frac{\delta^2}{64\Delta(\delta\Delta+1)}|V_0|=\frac{\delta^2}{32\Delta^3(\delta\Delta+1)}|Q|
\]
since $|Q|=|V_0|\Delta^2/2$. We have therefore proved that under the hypothesis
of Theorem~\ref{thm:main}, the decoder always finds a decomposition of $Z$.

Next we evaluate an upper bound on the Hamming weight $|\hat{R}_0+\hat{C}_1|$ of
the output of the decoder. To this end we upper bound the increment of
$|\hat{R}_0+\hat{C}_1|$ during sequential decoding, during all the first
parallel decoding steps, and finally during all second parallel decoding steps.

We note that during a sequential decoding step, the $Q$-neighbourhood of a
single vertex is modified, this vertex being in $V_{00}$ or in $V_{11}$. This
translates either into a modification of $\hat{R}_0$ or into a modification of
$\hat{C}_1$, which affects at most $\Delta$ rows or $\Delta$ columns.
Therefore a sequential decoding step translates into an augmentation of
$|\hat{R}_0+\hat{C}_1|$ by at most $\Delta^2$. Since the value of $|\hat{Z}|$
decreases at every step, the total contribution of sequential decoding does not
exceed $|Z|\Delta^2\leq 4|\error|\Delta^2$.

During a first parallel decoding step, the only modification to
$|\hat{R}_0+\hat{C}_1|$ is induced by the parallel decoder's action on vertices
$v$
of $V_{01}$. We note that the decoder modifies the associated $Q$-neighbourhood
only if it decreases $|\hat{Z}|$ on a subset of $Q(v)$ by at least
$w/2=\Delta^{3/2-\eps}/2$. Since all these sub-$Q$-neighbourhoods are disjoint,
their number is at most $2|\hat{Z}|/\Delta^{3/2-\eps}$ and the total increment
of $|\hat{R}_0+\hat{C}_1|$ during this step is at most $\Delta^2$ times this
number of vertices, namely $2|\hat{Z}|\Delta^{1/2+\eps}$. Now by
Proposition~\ref{prop:stalled} and Theorem~\ref{thm:hatZ}, the next time a
first parallel decoding step occurs the value of $|\hat{Z}|$ will have been
divided by at least $\kappa\Delta^{1/2-5\eps}$ for some constant $\kappa$, therefore the
sum of the contributions to $|\hat{R}_0+\hat{C}_1|$ of all parallel decoding
steps is at most $4|Z|\Delta^{1/2+\eps}\leq 16|\error|\Delta^{1/2+\eps}$ by summing the converging
geometric series.

Finally, for the contribution of the second parallel decoding steps, we argue as
before that the increment of $|\hat{R}_0+\hat{C}_1|$ is at most $\Delta^2$ per
vertex on which the decoder takes action. The number of such vertices is not
more than $|S_0'|+|S_0''|$, $|S_0'|$ being an upper bound on the number of
active vertices on which the decoder may act, and $|S_0''|$ being an upper bound
on the number of inactive vertices which will become active because the decoder
decides to update them. Applying Theorem~\ref{thm:tweaked} and Lemma~\ref{lem:S''00}
(which also holds for $S_{11}''$), we have that $|S_0'|\leq
\kappa|S_0|/\Delta^{1-4\eps}$ and $|S_0''|\leq \kappa|S_0|/\Delta^{1-4\eps}$ for some
constant $\kappa$. The increment to $|\hat{R}_0+\hat{C}_1|$ is therefore at most
\[
\kappa|S_0|\Delta^{1+4\eps}\leq \kappa\|\hat{Z}\|\Delta^{1+4\eps}
\]
for some constant $\kappa$, where $\hat{Z}$ refers to the value just before the
parallel decoding procedure is applied. As stated above, the maximum value of
$\|\hat{Z}\|$ is bounded from above by $8|\error|(\Delta+1/\delta)$. Since every
parallel decoding procedure divides the current value of $\|\hat{Z}\|$ by at
least $2$ (by Lemma~\ref{lem:/2}), it is an easy exercice in combinatorics to
see that the sum of all the values of $\|\hat{Z}\|$ before each parallel
decoding procedure cannot exceed $16|\error|(\Delta+1/\delta)$, irrespective of the
distribution of the parallel decoding steps during the whole decoding procedure.
We obtain therefore that total contribution of the second parallel decoding steps is at
most
\[
\kappa|\error|\Delta^{2+4\eps}
\]
for some constant $\kappa$ and $\Delta$ large enough.

Summarising, the upper bound on the contribution of the second parallel decoding
steps dominates the other terms and we have that the final output of the decoder
is $\hat{R}_0+\hat{C}_1$ of Hamming weight
\[
|\hat{R}_0+\hat{C}_1| \leq \kappa|\error|\Delta^{2+4\eps}
\]
for some constant $\kappa$. From \eqref{eq:<2|e|} we therefore obtain
\[
|\error+\hat{\error}|=|R_0+C_1+\hat{R}_0+\hat{C}_1|\leq
|R_0+C_1|+|\hat{R}_0+\hat{C}_1|\leq c|\error|\Delta^{2+4\eps}
\]
for some constant $\kappa$. Applying the lower bound on $d_{\min}$ from Theorem~\ref{thm:good-qLDPC}
we obtain that $|\error+\hat{\error}|<d_{\min}$ which concludes the proof. 
\end{proof}

%%%%%%%%%%%%%%%%%%%%%%%%%%%%
%%%%%%%%%%%%%%%%%%%%%%%%%%%%
%%%%%%%%%%%%%%%%%%%%%%%%%%%%
%%%%%%%%%%%%%%%%%%%%%%%%%%%%

\section{Link with the Lifted Product codes}
\label{sec:LP}

In this section, we briefly recall the construction of the lifted product codes introduced by Panteleev and Kalachev~\cite{PK21} and explain how quantum Tanner codes can be obtained from lifted product codes, or more generally from hypergraph product codes. Finally, we show that the decoder we studied in the previous sections immediately gives an efficient decoder for lifted product codes.

\subsection{Construction of the Lifted Product codes}

Take a group $G$ and two symmetric subsets $A = \{a_i\} \subset G$, $B = \{b_i\} \subset G$ of size $\Delta$. 
Let $h_A$ and $h_B$ be two parity-check matrices on $\F_2 \subset \F_2[G]$ of size $m \times \Delta$.
We further define two diagonal matrices of size $\Delta$ over $G$:
\[D_A := \mathrm{diag}(a_1, \ldots, a_\Delta), \qquad D_B := \mathrm{diag} (b_1, \ldots, b_\Delta)\]
and set $H_A := h_A D_A$ and $H_B := h_B D_B$. 

Consider the graph with two vertices and $\Delta$ parallel edges, and define two Tanner codes $T_A$ and $T_B$ on this graph, where the local constraints on both vertices are given by the codes $\ker h_A$ and $\ker h_B$, respectively. The corresponding Tanner graphs have vertex sets $A \cup (C_0 \cup C_1)$ for $T_A$ and $B \cup (D_0 \cup D_1)$ for $T_B$. 
In particular, we have $|C_0| = |C_1| = |D_0| = |D_1| = m$.

The lifted product codes of~\cite{PK21} are defined on qubits indexed by ($|G|$ copies of) $A\times B \cup (C_0 \cup C_1) \times (D_0 \cup D_1)$, and have $\sigma_X$-generators indexed by ($|G|$ copies of) $A \times (D_0 \cup D_1)$, and $\sigma_Z$-generators by ($|G|$ copies of) $(C_0 \cup C_1) \times B$.  
We will use the following identification between $\F_2^G$ and the group algebra $\F_2[G]$: 
\[ (x_g)_{g \in G} \quad \leftrightarrow \quad \sum_{g \in G} x_g g\]
where $x_g \in \F_2$. This allows us to describe a subset of $A\times B \times G$, \textit{e.g.}\ the support of a stabilizer, as a matrix in $\F_2[G]^{A\times B}$.

A relatively simple way of describing the lifted product codes is \emph{via} their stabilizer elements. 
For instance, a general $\sigma_Z$-stabilizer is associated with two matrices $(U_0, U_1)$ with entries indexed by $C_0 \times B$ and $C_1 \times B$, and values in $\F_2[G]$.
The corresponding operator is a product of $\sigma_Z$-Pauli matrices with a support given by
\[ \underbrace{ h_A^T U_0 + H_A^T U_1}_{A \times B} + \underbrace{U_0 h_B^T}_{C_0\times D_0} + \underbrace{ U_0 H_B^T}_{C_0\times D_1} + \underbrace{U_1 h_B^T }_{ C_1\times  D_0}+ \underbrace{ U_1 H_B^T}_{ C_1\times D_1}.\]
We leave the factor $G$ implicit for ease for notation, and represent the support of a set by its indicator vector.
Similarly, a general $\sigma_X$-stabilizer element $(V_0, V_1)$ has support
\[ \underbrace{V_0 h_B + V_1 H_B}_{A \times B} + \underbrace{h_A  V_0}_{C_0 \times D_0} + \underbrace{H_A V_0}_{ C_1\times D_0} + \underbrace{h_A V_1 }_{C_0  \times D_1}+\underbrace{ H_A V_1}_{ C_1 \times D_1}.\] 
A generator of the code is obtained by taking $(U_0, U_1)$, or $(V_0, V_1)$, of Hamming weight equal to 1. The generators of the lifted product codes have weight at most $2\Delta$. 

\subsection{Quantum Tanner codes from Lifted Product codes}

Let us define matrices $g_A$ indexed by $C' \times A$ and $g_B$ by $D' \times B$ to be the parity-check matrices of $(\ker h_A)^\perp$ and $(\ker h_B)^\perp$. Note that $|C'| = |D'| = \Delta - m$ and that 
\[ g_A h_A^T=0, \quad g_B h_B^T=0.\]
We further define $G_A := g_A D_A^{-1}$ and $G_B := g_B D_B^{-1}$ that satisfy $G_A H_A^T=0$ and $G_B H_B^T=0$.

The quantum Tanner code is nothing but the code with $\sigma_Z$-stabilizers $(U_0', U_1') \in \F_2[G]^{(C_0 \cup C_1) \times D'}$ and $\sigma_X$-stabilizers $(V_0', V_1') \in \F_2[G]^{(C' \times (D_0 \cup D_1)}$ corresponding to the stabilizers $(U_0, U_1) := (U_0' G_B, U_1' g_b)$ and $(V_0, V_1) := (g_A^T V_0', G_A^T V_1')$ of the Lifted Product code. 
They take the following form:
\begin{align*}
\sigma_Z-\text{stabilizer} \; (U_0', U_1')  &: \quad \underbrace{ h_A^T U_0' G_B+ H_A^T U_1' g_b}_{A \times B} + \underbrace{U_0' G_B h_B^T}_{C_0\times D_0} + \underbrace{ U_0' G_B H_B^T}_{C_0\times D_1} + \underbrace{U_1' g_B h_B^T }_{ C_1\times  D_0}+ \underbrace{ U_1' g_b H_B^T}_{ C_1\times D_1}\\
& \quad = 
\underbrace{ h_A^T U_0' G_B+ H_A^T U_1' g_b}_{A \times B} + \underbrace{U_0' G_B h_B^T}_{C_0\times D_0}  +  \underbrace{ U_1' g_b H_B^T}_{ C_1\times D_1},\\
\sigma_X-\text{stabilizer}\; (V_0', V_1')  &: \quad \underbrace{g_A^T V_0' h_B + G_A^T V_1' H_B}_{A \times B} + \underbrace{h_A  g_A^T V_0'}_{C_0 \times D_0} + \underbrace{H_A g_A^T V_0'}_{ C_1\times D_0} + \underbrace{h_A G_A^T V_1' }_{C_0  \times D_1}+\underbrace{ H_A G_A^T V_1'}_{ C_1 \times D_1}\\
&\quad =\underbrace{g_A^T V_0' h_B + G_A^T V_1' H_B}_{A \times B} + \underbrace{H_A g_A^T V_0'}_{ C_1\times D_0} + \underbrace{h_A G_A^T V_1' }_{C_0  \times D_1}.
\end{align*}
We immediately notice that the supports of the $\sigma_X$-stabilizers and the $\sigma_Z$-stabilizers only overlap on $A\times B$. In particular, we can simply discard all the qubits in $(C_0 \cup C_1) \times (D_0 \cup D_1)$ and define a shorter stabilizer code on $A \times B$ with the stabilizers corresponding to $(U_0', U_1')$ and $(V_0', V_1')$. This is nothing but the quantum Tanner code $\Q = (\C_0, \C_1)$ with local tensor codes $C_A \otimes C_B$ and $C_A^\perp \times C_B^\perp$ where we define
\[ C_A = \ker g_A, \quad C_B = \ker h_B, \quad C_A^\perp = \ker h_A, \quad C_B^\perp = \ker g_B.\]

\subsection{Decoding Lifted Product codes}

An error vector $\error = (\error^X, \error^Z)$ for the lifted product code is described a $\sigma_X$-type error $\error^X$ and a $\sigma_Z$-type error $\error^Z$:
\[ \error^{X/Z} = \error^{X/Z}_{AB} + \error^{X/Z}_{00} + \error^{X/Z}_{01} + \error^{X/Z}_{10} + \error^{X/Z}_{11}\]
with $\error^{X/Z}_{AB} \in \F_2[G]^{A\times B}$ and $\error^{X/Z}_{ij} \in \F_2[G]^{C_j \times D_i}$.
This error gives rise to a syndrome $(S_0, S_1,T_0, T_1)$ with $(S_0, S_1)$ detected by $\sigma_Z$-generators and $(T_0, T_1)$ detected by $\sigma_X$-generators, as follows:
\begin{align*}
S_0 & =  h_A \error^{Z}_{AB}  + \error^{Z}_{00} h_B + \error^{Z}_{10} H_B  \quad \in \F_2[G]^{C_0 \times B},  \\
S_1& = H_A \error^{Z}_{AB} + \error^{Z}_{01} h_B + \error^{Z}_{11} H_B  \quad \in \F_2[G]^{C_1\times B},\\
T_0 & =  \error^{X}_{AB} h_B^T + h_A^T \error^{X}_{00} + H_A^T \error^{X}_{01} \quad \in \F_2[G]^{A \times D_0},\\
T_1& =  \error^{X}_{AB} H_B^T + h_A^T \error^{X}_{10}  + H_A^T \error^{X}_{11} \quad \in \F_2[G]^{A \times D_1}.
\end{align*}

One of the features of this quantum code (and also of the hypergraph product code of two classical Tanner codes) is that one can easily remove the error from $C_0 \times D_1 \cup C_1 \times D_0$ by adding a stabilizer of weight $O(|\error_{10} \cup \error_{01}|)$. Since the code is LDPC, it means that this can be done without increasing the error weight too much. 
\begin{lemma}\label{lem:simplification}
For any error $\error$, there exists an equivalent error $\error'$ such that $\error'_{01} = 0$, $\error'_{10}=0$ and 
\[ |\error'| \leq |\error| + 4\Delta^2 (|\error_{01}| + |\error_{10}|).\]
\end{lemma}

\begin{proof}
Consider a general error $(\error^X, \error^Z)$. We focus on the $\sigma_Z$-part $(\error_{AB}^Z, \error_{00}^Z, \error_{10}^Z, \error_{01}^Z, \error_{11}^Z)$. The $\sigma_X$-part is treated identically.
Let $U_0, U_1$ be matrices such that 
\[ U_0 H_B =\error_{10}^Z, \quad U_1 h_B = \error_{01}^Z.\]
We have that $|U_0| \leq \Delta |\error_{10}^Z|$ and $|U_1| \leq \Delta |\error_{01}^Z|$.
These stabilizers induce new errors $h_A^T U_0 + H_A^T U_1 + U_0 h_B^T+U_1 H_B^T$, which have a weight at most $2\Delta (|U_0| + |U_1|) \leq 2\Delta^2( |\error_{10}|+ |\error_{01}|)$.
\end{proof}

Thanks to Lemma~\ref{lem:simplification}, we now focus on a $\sigma_Z$-error $\error$ such that $\error_{01}=0$ and $\error_{10}=0$. Its syndrome reads:
\[T_0 = \error_{AB} h_B^T + h_A^T \error_{00}, \qquad T_1 = \error_{AB} H_B^T +H_A^T \error_{11}.\]
We can multiply these syndromes by $g_A$ and $G_A$ respectively and obtain
\[ g_A T_0 = g_A \error_{AB} h_B^T,\qquad G_A T_1 = G_A \error_{AB} H_B^T,\]
which is exactly the decoding problem for quantum Tanner codes. We note that a similar strategy was already exploited in~\cite{QC21}. 
The decoder presented in the previous sections can correct this error provided its weight is small enough (Theorem~\ref{thm:main}). 
More precisely, if $|\error_{AB}| \leq \kappa n/\Delta^4$, then the decoder outputs an equivalent error $\hat{\error}_{AB}$ for the quantum Tanner code. In other words, there exists a stabilizer element associated with $(U_0, U_1)$ such that 
\[ \hat{\error}_{AB} = \error_{AB} +h_A^T U_0' G_B + H_A^T U_1' g_B,\]
as described in the previous subsection.

Let us introduce a matrix $n_A$ indexed by $C_0 \times A$ such that $n_A h_A^T = \1_{C_0}$. We assume here that $h_A$ is a full-rank matrix. We also denote $N_A := n_A D_A^{-1}$, which implies $N_A H_A^T = \1_{C_1}$.

The decoder for the lifted product code will simply return the triple $(\hat{\error}_{AB}, \hat{\error}_{00}, \hat{\error}_{11})$ with 
\[ \hat{\error}_{00} := n_A (T_0 + \hat{\error}_{AB} h_B^T), \quad \hat{\error}_{11} := N_A (T_1 + \hat{\error}_{AB} H_B^T).\]
Straightforward manipulations show that
\begin{align*}
 \hat{\error}_{00}  &= n_A (( \error_{AB} h_B^T + h_A^T \error_{00}) +(\error_{AB} +h_A^T U_0' G_B + H_A^T U_1' g_B) h_B^T)  \\
&=  \error_{00} + U_0' G_B  h_B^T  
\end{align*}
and 
\begin{align*}
 \hat{\error}_{11}  &= N_A (( \error_{AB} H_B^T + H_A^T \error_{11}) +(\error_{AB} +h_A^T U_0' G_B + H_A^T U_1' g_B) H_B^T)  \\
&=  \error_{11} +  U_1' g_B H_B^T
\end{align*}
which means that the true error and the error returned by the decoder differ by
\[ \underbrace{ h_A^T (U_0' G_B) + H_A^T (U_1' g_B)}_{A \times B} + \underbrace{ (U_0' G_B) h_B^T}_{C_0\times D_0} +  \underbrace{ (U_1' g_B)H_B^T}_{ C_1\times D_1}.\]
In other words, they differ by the stabilizer element $(U_0, U_1) = (U_0' G_B, U_1' g_B)$ of the lifted product code, meaning that they are equivalent and that the decoder succeeded. \\

We have therefore proved the following:
\begin{theo}\label{thm:PK}
There exists a constant $\kappa$, depending only on $\delta$, a lower bound for the
minimum distances of both component codes $C_A$ and $C_B$, such that for large
enough fixed $\Delta$, the above decoding algorithm corrects all error patterns
of weight less than $\kappa n/\Delta^6$ for the lifted product code of length
$n=|Q|$.
\end{theo}

\newpage
%\bibliographystyle{alpha}
%\bibliography{biblio}

\newcommand{\etalchar}[1]{$^{#1}$}

\end{document}